\documentclass[showpacs,twocolumn,pra,superscriptaddress,notitlepage]{revtex4-1}
\usepackage{color}
\usepackage{amsmath,amsthm,mathrsfs,dsfont}
\usepackage{graphicx}
\usepackage{tikz}
\usepackage{xr}
\usepackage{xc}
\usepackage{float}
\usepackage{bbm}
\usepackage{epsfig} 
\usepackage{verbatim}
\usepackage{array}
\usepackage{setspace}
\usepackage{bbding}
\usepackage{amssymb}
\usepackage{pifont}
\usepackage{braket}
\usepackage{tikz}

\usepackage{multirow}
\usepackage{setspace}
\usetikzlibrary{quantikz}

\usepackage{newtxtext,newtxmath}
\usepackage{lineno}

\externalcitedocument[suppl-]{../npjqiv3suppl/npjqiv3suppl}
   
\usepackage[unicode=true,
 bookmarks=true,bookmarksnumbered=false,bookmarksopen=false,
 breaklinks=false,pdfborder={0 0 1},backref=false,colorlinks=true]
 {hyperref}
\hypersetup{
 linkcolor=magenta, urlcolor=blue, citecolor=blue, pdfstartview={FitH}, hyperfootnotes=false, unicode=true}
 
\usepackage{color}
\newtheorem{lemma}{Lemma}
\newtheorem{theorem}{Theorem}

\makeatletter
\@ifundefined{textcolor}{}
{%
 \definecolor{BLACK}{gray}{0}
 \definecolor{WHITE}{gray}{1}
 \definecolor{RED}{rgb}{1,0,0}
 \definecolor{GREEN}{rgb}{0,1,0}
 \definecolor{BLUE}{rgb}{0,0,1}
 \definecolor{CYAN}{cmyk}{1,0,0,0}
 \definecolor{MAGENTA}{cmyk}{0,1,0,0}
 \definecolor{YELLOW}{cmyk}{0,0,1,0}
}

\usepackage{amsfonts}\usepackage{tabularx}\usepackage{dcolumn}\usepackage{bm}\usepackage{graphicx}\usepackage{epstopdf}

\hypersetup{urlcolor=blue}

\makeatother

\newcolumntype{C}[1]{>{\centering\arraybackslash$}p{#1}<{$}}

\usepackage{algorithm}
\usepackage{algorithmic} 

\begin{document}

\widetext

\title{Circuit complexity of quantum access models for encoding classical data}

\author{Xiao-Ming Zhang}
\email{phyxmz@gmail.com}
\affiliation {Key Laboratory of Atomic and Subatomic Structure and Quantum Control (Ministry of Education), Guangdong Basic Research Center of Excellence for Structure and Fundamental Interactions of Matter, School of Physics, South China Normal University, Guangzhou 510006, China} 
\affiliation {Guangdong Provincial Key Laboratory of Quantum Engineering and Quantum Materials, Guangdong-Hong Kong Joint Laboratory of Quantum Matter, Frontier Research Institute for Physics, South China Normal University, Guangzhou 510006, China}
\affiliation{Center on Frontiers of Computing Studies, Peking University, Beijing 100871, China}
\affiliation{School of Computer Science, Peking University, Beijing 100871, China}

\author{Xiao Yuan}

\affiliation{Center on Frontiers of Computing Studies, Peking University, Beijing 100871, China}
\affiliation{School of Computer Science, Peking University, Beijing 100871, China}

\begin{abstract}
\noindent \textbf{Abstract} How to efficiently encode classical data is a fundamental task in quantum computing. 
While many existing works treat classical data encoding as a black box in oracle-based quantum algorithms, their explicit constructions are crucial for the efficiency of practical algorithm implementations. Here, we unveil the mystery of the classical data encoding  black box and study the Clifford$+T$ complexity in constructing several typical quantum access models. For general matrices (even including sparse ones), we prove that sparse-access input models and block-encoding both require nearly linear circuit complexities relative to the matrix dimension. We also give construction protocols achieving near-optimal gate complexities.
On the other hand, the construction becomes efficient with respect to the data qubit when the matrix is a linear combination of polynomial terms of efficiently implementable unitaries. As a typical example, we propose improved block-encoding when these unitaries are Pauli strings. Our protocols are built upon improved quantum state preparation and a select oracle for Pauli strings, which hold independent values. Our access model constructions provide considerable flexibility, allowing for tunable ancillary qubit numbers and offering corresponding space-time trade-offs.
\end{abstract}
\maketitle

\section{Introduction}\label{sec:intro}

Quantum computing offers speedups over the classical counterpart in different tasks, including factoring, searching, simulation, etc~\cite{Nielsen.02}. 
However, the speedups, in many cases, rely on the existence of efficient oracles or access models to encode the related classical data~\cite{Aaronson.21}. 
In this context, a function $f(x)$ representing the classical data of interest is encoded using a unitary operation $U_f$, which acts as an oracle in the computation. To study quantum advantages, the number of queries to $U_f$ in a quantum algorithm is compared to the number of queries to $f(x)$ in classical algorithms. Quantum computing provides substantial reduction in query complexity for many problems of practical importance~\cite{Grover.97,Berry.07,Harrow.09}. 
 
There are various access models to encode classical data. One commonly used access model is the sparse-access input model (SAIM)~\cite{Berry.07,Childs.11,Childs.10,Berry.14,Harrow.09,Childs.17,Gilyen.19,Chakraborty.19,Babbush.23}, which encodes general sparse matrices and outputs the value or position of the non-zero elements when provided with appropriate inputs. SAIM is initially introduced for Hamiltonian simulation and discrete quantum walks~\cite{Berry.07,Childs.11,Childs.10,Berry.14}, and has then found broad applications in other fields such as machine learning~\cite{Harrow.09,Childs.17,Chakraborty.19} and classical oscillator simulations~\cite{Babbush.23}. For example, the quantum linear system problem could be solved with $\tilde O(\kappa)$ queries to SAIM~\cite{Harrow.09,Childs.17,Gilyen.19}, where $\kappa$ represents the condition number of the matrix to be inverted.

Another important access model is block-encoding, which serves as a crucial subroutine for quantum signal processing~\cite{Low.17,Low.19} and its generalization~---~quantum singular-value transformations (QSVT)~\cite{Gilyen.19,Martyn.21}. 
The success of
block-encoding enables the realization of Hamiltonian simulation with an optimal query complexity~\cite{Low.17,Low.19}. Furthermore, many seminal quantum algorithms, including Grover's algorithm, quantum Fourier transformation, and the HHL algorithm, could be viewed as special cases of QSVT, where the problem of interest is encoded using block-encoding~\cite{Martyn.21}.

Many existing works treat access models as black boxes for convenience. However, the actual circuit complexity of the algorithm also depends on the cost of each query to these access models. While being important, this problem only draws much attention very recently with many basic problems still left open. In particular, Ref~\cite{Clader.22} presents a nearly time-optimal protocol for block-encoding of general dense matrices of $2^n\times 2^n$ dimension. A circuit depth of $\tilde O(n)$ can be achieved at the expense of exponential ancillary qubits. Ref~\cite{Sunderhauf.23} examines matrices with $D$ data each appearing $M$ times and considers examples including checkerboard matrices and tridiagonal matrices with polynomial circuit complexities. However, the cost of block-encoding of more general matrices remains unexplored. Moreover, it is still unclear if there is a fundamental limit to the resource required by data encoding.  

In this work, we provide a framework of constructing quantum access models in the fault-tolerant setting using Clifford$+T$ gates. The protocol works for general classical data and takes the underlying structure of the data, such as sparsity and linear combintaion of unitaries (LCU), into consideration. Our results represent a direct mapping from the query complexity of quantum algorithms to their practical circuit complexity. Our protocols allow tunable ancillary qubit numbers and offer space-time trade-off. 
For general sparse matrices of dimension $2^n=N$, we investigate the SAIM and block-encoding. 
For both access models, we first show that the gate count lower bound increases about linearly with respect to $N$. We then develop construction algorithms with varying ancillary qubit numbers ranging from $\Omega(n)$ to $O(N)$. Across the entire range of qubit numbers, we achieve nearly optimal circuit complexity. 
We next study the block-encoding of LCU. Efficient block-encoding is achievable when the matrix can be represented as a linear combination of a polynomial number of unitaries, which  can be implemented using  polynomial-size quantum circuits.

Our access model construction relies on optimized realizations of various subroutines that are independently valuable, including quantum state preparation, selective oracles for Pauli strings, and sparse Boolean functions. In all the listed operations, we achieve improved or at least comparable circuit complexities compared to the best-known realizations.

We now introduce the definition of SAIM and block-encoding in below. 
Let $N=2^n$, we consider a sparse matrix $H\in\mathbb{C}^{N\times N}$ with at most $s=O(1)$ nonzero elements at each row and column. Let $H_{x,y}$ be the value of the element at the $x$th row and $y$th column, and each $H_{x,y}$ is a $d$-digit integer ($d=O(1)$). Let idx denote a $2n$-qubit index register, and wrd denote an $n$-qubit word register, the sparse-access input model (SAIM) corresponds to two unitaries $O_H$, $O_F$, which satisfies
\begin{subequations}\label{eq:ofh}
\begin{align}
O_H|x,y\rangle_{\text{idx}}|z\rangle_{\text{wrd}}&=|x,y\rangle_\text{idx}|z\oplus H_{x,y}\rangle_\text{wrd},\\
O_F|x,k\rangle_{\text{idx}}&=|x,F(x,k)\rangle_{\text{idx}}.
\end{align}
\end{subequations}
Here, $F(x,k)$ is the column index of the $k$th nonzero element in row $x$. Due to its simplicity and generality, Eq.~\eqref{eq:ofh} becomes one of the standard access models in quantum computing, which is usually assumed to be available in processing classical data.

\noindent We call a unitary $U$ the block encoding of $H$ if we have  $$\alpha\left(\langle0^{n_{\text{anc}}}|\otimes\mathbb{I}_{N}\right)U\left(|0^{n_{\text{anc}}}\rangle\otimes\mathbb{I}_{N}\right)=H,$$ where $\alpha>0$ is the normalization factor, $n_{\text{anc}}$ is the number of  ancillary qubits, and $\mathbb{I}_{N}$ is the $N$-dimensional identity. In practice, we may consider approximated construction of the block encoding. More specifically, we call unitary $\tilde U$ an $(\alpha,n_{\text{anc}},\varepsilon)-$block-encoding of $H$ if 
\begin{align}\label{eq:bl}
\left\|H-\alpha\left(\langle0^{n_{\text{anc}}}|\otimes\mathbb{I}_{N}\right)\tilde U\left(|0^{n_{\text{anc}}}\rangle\otimes\mathbb{I}_{N}\right)\right\|\leqslant\varepsilon
\end{align}
for error parameter $\varepsilon\ge 0$.
Throughout our manuscript, $\|\cdot\|$ represents either the spectral norm for matrices or Euclidean norm for vectors. For a general $N$-dimensional matrix $H$, the construction of its block-encoding requires $\Omega(\text{Poly}(N))$ gate count. This is true even for sparse $H$ as we show in Supplementary Discussion 2.

On the other hand, when $H$ has some other structures, the resource may be significantly reduced. In particular, we consider $H$ in the form of a linear combination of unitaries (LCU) as
\begin{align}\label{eq:lcugen}
H=\sum_{p=0}^{P-1}\alpha_pu_p,
\end{align}
where $u_p$ are $n$-qubit unitaries that can be implemented with polynomial-size quantum circuit, and $P=O(\text{poly}(n))$. The concept ``LCU'' appeared firstly in~\cite{Long.06}. The main purpose of Ref.~\cite{Long.06} and the follow-up work Ref.~\cite{Long.11} is to realize non-unitary transformation on quantum computers. In the context of Hamiltonian simulation, Ref.~\cite{Childs.12} has shown that LCU-based method can outperform product formula based methods. Many subsequent works with different applications have then been inspired~\cite{Low.19,Cong.19,Martyn.21,Wei.22,HaiSheng.22}. 

Without loss of generality, we may assume that $\log_2P$ is an integer, and $\sum_{p=0}^{P-1}\alpha_p=1$. This can always be satisfied by adding terms with zero amplitude, and rescaling the Hamiltonian. In particular, the linear combination of Pauli strings
\begin{align}\label{eq:lcu}
H=\sum_{p=0}^{P-1}\alpha_pH_p
\end{align}
will be studied in details. Here, $\alpha_p>0$, $P\geqslant1$, $H_p=\bigotimes_{l=1}^nH_{p,l}$, and $H_{p,l}\in\{\pm I,\pm X,\pm Y,\pm Z\}$ are single-qubit Pauli operators. Eq.~\eqref{eq:lcu} is important as it corresponds to the Hamiltonian of almost all physical quantum systems, such as the spin and molecular systems.    

In our constructions, we consider the fault-tolerant quantum computing setting. More specifically, we only use two-qubit Clifford gate and single-qubit $T$ gate, which is equivalent to the elementary gate set $\mathcal{G}_{\text{clf}+T}\equiv\{H,S,T,\text{CNOT}\}$. All gates in $\mathcal{G}_{\text{clf}+T}$ are error-correctable with surface code~\cite{Fowler.12}. We benchmark the circuit complexity of a given quantum circuit with three quantities: total number of elementary gates, total qubit number (including data qubits and ancillary qubits), and circuit depth. We will also discuss the space-time trade-off of our algorithm, i.e. the circuit depth under a certain number of ancillary qubits. We also allow at least $O(n)$ ancillary qubits, because this does not increase the total space complexity.

\section{Results}

\noindent\textbf{Circuit complexity lower bound}

\noindent Before discussing the access model construction, we first study the lower bound of the circuit complexity.  We will focus on the encoding of sparse matrices. The methodology here is general and can be readily applied to other related problems.

Our strategy is as follows. Firstly, we analyze the \textit{capacity} of a quantum circuit with bounded resource, i.e. how much unique unitaries can be constructed, given fixed number of elementary gates or circuit depth. Secondly, we analyze the \textit{size} of the access model, i.e. the number of unique unitaries required to approximate the access model with arbitrary parameters. The circuit complexity can then be estimated by comparing the capacity of a quantum circuit and the size of the access model. All proofs of our lemma and theorems in this section are provided in Supplementary Discussion 1.

$\\$
\noindent\textit{Quantum circuit capacity.} 
\noindent Assuming that we are given a finite two-qubit elementary gate set $\mathcal{G}_{\text{ele}}$. We define $g\equiv|\mathcal{G}_{\text{ele}}|=O(1)$ with $|\cdot|$ the number of elements in the set. Our first result is that the capacity can be lower bounded only with the number of elementary gates, independent of the space and time resources. 

\begin{lemma}\label{lm:CC}
Let $\mathcal{G}_{C}$ be the set containing all $n$-qubit unitaries that can be constructed with $C$ elementary gates in $\mathcal{G}_{\text{ele}}$. Then, we have $\log|\mathcal{G}_{C}|= O\left((C\log(C+n))\right)$, even with unlimited ancillary qubit number. 
\end{lemma} 

Lemma.~\ref{lm:CC} implies that the capacity does not always increase with ancillary qubit number, which can be understood as follows. All ancillary qubits should be uncomputed at the end of the circuit. When $C$ is fixed, only finite number of unitaries can satisfy this requirement, while constructable by those elementary gates. We also note that the circuit depth $D$ is bounded by $C$, so Lemma.~\ref{lm:CC} also implies a relation between capacity and circuit depth. 

On the other hand, when the ancillary qubit number and circuit depth are finite, the lower bound of capacity can be tighten as follows.

\begin{lemma}\label{lm:cancd}
Let $\mathcal{G}'_{n_{\text{anc}},D}$ be the set containing all unitaries that can be constructed with $n_{\text{anc}}$ ancillary qubits and $D$ circuit depth. Then, we have $\log\big|\mathcal{G}'_{n_{\text{anc}},D}\big|=O\left(D(n+n_{\text{anc}})\right)$.
\end{lemma} 

Lemma.~\ref{lm:CC},~\ref{lm:cancd} represent the ultimate representational power of quantum circuits constructed with local gates. Lemma.~\ref{lm:CC} and \ref{lm:cancd} can be used to estimate the circuit complexity lower bound whenever the tasks have requirement on $|\mathcal{G}_C|$ or $|\mathcal{G}_{n_\text{anc},D}'|$.  Moreover, similar results can be obtained straightforwardly for other type of elementary gate sets, such as $k$-local operations with $k>2$. 

\noindent\textit{Circuit complexity for encoding sparse matrices.}
\noindent With Lemma.~\ref{lm:CC} and \ref{lm:cancd}, we now estimate the circuit complexity lower bound for accessing sparse matrices. 
For SAIM, it turns out that at least $\Omega(N!)$ unique unitaries are required to cover the set of all SAIM for $1$-sparse matrices. So according to Lemma.~\ref{lm:CC},~\ref{lm:cancd}, we have the following result.

\begin{theorem}
\label{th:lbsm}
Given an arbitrary finite two-qubit elementary gate set $\mathcal{G}_{\text{ele}}$.
Let $n_{\text{anc}}$, $D$ and $C$ be the number of ancillary qubits, circuit depth and total number of gates in $\mathcal{G}_{\text{ele}}$ required to approximate SAIM in Eq.~\eqref{eq:ofh} with any accuracy $\varepsilon<1$. Then, we have $(n+n_{\text{anc}})D=\Omega(2^nn)$ and $C=\Omega(2^{n})$.
\end{theorem}

A similar result is also obtained for the block-encoding of sparse matrix as follows.

\begin{theorem}\label{th:hd_spbe}
Given an arbitrary finite two-qubit elementary gate set $\mathcal{G}_{\text{ele}}$.
Let $n_{\text{anc}}$, $D$ and $C$ be the number of ancillary qubits, circuit depth and total number of gates in $\mathcal{G}_{\text{ele}}$ required to construct the block-encoding of $H$ with any accuracy $\varepsilon<2$.
 Then, we have $(n+n_{\text{anc}})D=\Omega(N)$ and $C=\Omega(N^\alpha)$ for arbitrary $\alpha\in(0,1)$.
\end{theorem}

Theorem.~\ref{th:lbsm},~\ref{th:hd_spbe} imply that a general sparse matrix can not be encoded with subexponential quantum gates, for both SAIM and block-encoding.
It is possible to trade ancillary qubit numbers for the circuit depth. However, the space and time complexities can not achieve sub-exponential scaling simultaneously. The hardness of SAIM can be interpreted as follows. Although $H$ is assumed to be sparse ($O(1)$ nonzero elements at each row and column), there are still totally $2^n\times O(1)=O(2^n)$ number of independent variables in total. Therefore, the quantum circuit should be large enough to contain exponential number of elementary gates.

We note that the quantum circuits capacity for ancillary-free case has been studied in Section 4.5.4 of~\cite{Nielsen.02}. Moreover, a related result to Theorem.~\ref{th:lbsm} has obtained in~\cite{Jaques.23}, which gives a distinct quantum circuit number lower bound with fixed qubit number, and show that there exists a table of size $N$ requiring $\Omega(N)$ gate count. Ref.~\cite{Nielsen.02} allows approximated implementations, but does not consider ancillary qubit usage. Ref.~\cite{Jaques.23} implicitly allows ancillary qubits, but does not consider approximated implementations. On the contrary, our results are more general, because both ancillary qubit usage and approximated implementations are allowed. Our results can be generalized from unitary to quantum channels. In Supplementary Discussion 1, we show that the circuit capacity and circuit lower bound are similar if we consider two-qubit quantum channels as elementary quantum operations, which can include measurement and feedback controls.

 $\\$
\noindent\textbf{Quantum state preparation}

\noindent Quantum state preparation is a critical step of our access model construction and of independent interest. We say that a $(n+n_{\text{anc}})$ qubit unitary $G$ prepares the $n$-qubit quantum state $|\psi\rangle$ with accuracy $\varepsilon$ if 
\begin{align}
G(|0^n\rangle\otimes|0^{n_{\text{anc}}}\rangle)=|\tilde{\psi}\rangle\otimes|0^{n_{\text{anc}}}\rangle
\end{align}
for some $\left\||\psi\rangle-|\tilde{\psi}\rangle\right\|\leqslant\varepsilon$. 

Such a problem has been studied extensively~\cite{Long.01,Grover.02,Mottonen.05,Plesch.11,Zhicheng.21,Zhang.21,Sun.21,Rosenthal.21,Zhang.22,Clader.22,Yuan.22,Ashhab.22,Gui.23}. When given sufficiently large among of ancillary qubits, the optimal Clifford$+T$ depth $O(n+\log(1/\varepsilon))$ can be achieved~\cite{Gui.23}. However, with restricted ancillary qubit number, the optimal circuit depth has not been reached. For example, with $O(n)$ ancillary qubits, the best-known Clifford$+T$ construction has achieved $O((N/n)\log(N/\varepsilon))$ circuit depths~\cite{Sun.21}. Besides, for gate count scaling, all existing algorithms have either $O(N\text{poly}(n))$ or $O(N\text{polylog}(n))$ Clifford$+T$ count. It remains an outstanding question if a linear gate count scaling with respect to the data dimension $N$ can be reached. 

\begin{table}[h]
\caption{Clifford+T complexities of $n$-qubit state preparation protocols with fixed accuracy $\varepsilon$ and total qubit (data qubit $+$ ancillary qubit) number $O(n)$.
The $\varepsilon$ scaling of Clifford$+T$ count and depth are $O(\log(1/\varepsilon))$ for all protocols. The $\varepsilon$ scaling of qubit number is $O(\log(1/\varepsilon))$ for Ref~\cite{Low.18} and $O(1)$ for all other schemes. Ref~\cite{Low.18} also minimize $T$ complexities.
\label{tab:1}}
\begin{ruledtabular}
\begin{tabular}{ccccccc}
Protocols &Count &Depth\\
\hline
Ref~\cite{Long.01,Grover.02,Mottonen.05,Plesch.11}& $O\left(N\text{poly}(n)\right)$& $O\left(N\text{poly}(n)\right)$ \\
Ref~\cite{Low.18}& $O\left(N\log(n)\right)$& $O\left(N\right)$ \\
Ref~\cite{Sun.21}& $O\left(Nn\right)$& $O\left(N\right)$ \\
Theorem.~\ref{th:sp_main}&  $\bm{O\left(N\right)}$&$\bm{O\left(N\frac{\log n(\log\log n)}{n}\right)}$
\end{tabular}
\end{ruledtabular}
\end{table}

\begin{table}[h]
\caption{Clifford+T complexities of $n$-qubit state preparation protocols with fixed $\varepsilon$ and exponential ancillary qubits. The $\varepsilon$ scaling of Clifford$+T$ count and depth are $O(\log(1/\varepsilon))$ for all schemes. Total qubit numbers are $O(N\text{poly}(n))$ for   Ref~\cite{Rosenthal.21}, $O(N\log(1/\varepsilon))$ for Ref~\cite{Low.18}, and $O(N)$ for  all other protocols. Ref~\cite{Low.18} and Ref~\cite{Clader.22} also minimize $T$ complexities. Protocols labelled $^{*}$ only require sparse connectivity, i.e. each qubit connect to $O(1)$ of other qubits. 
\label{tab:2}
}
\begin{ruledtabular}
\begin{tabular}{ccccccc}
Protocols &Count& Depth\\
\hline
Ref~\cite{Sun.21,Rosenthal.21}& $O(Nn)$& $O(n^2)$\\
Ref~\cite{Low.18,Clader.22}& $O(N\log n)$& $O(n^2)$ \\
Ref~\cite{Zhang.22}$^{*}$&  $O(N\log n)$& $O(n\log n)$ \\
Ref~\cite{Gui.23}& $O(N\log n)$& $\bm{O(n)}$\\
Theorem.~\ref{th:sp_main}$^{*}$&  $\bm{O\left(N\right)}$ & $O\left(n\log n \right)$
\end{tabular}
\end{ruledtabular}
\end{table}

Here, we provide a family of improved quantum state preparation protocols with tunable ancillary qubit number. The result is summarized in below (follows directly from Theorem.~\ref{th:stsp} in Methods).

\begin{theorem}\label{th:sp_main}
With $n_{\text{anc}}$ ancillary qubits where $\Omega(n)\leqslant n_{\text{anc}}\leqslant O(N)$, an arbitrary $n$-qubit quantum state can be prepared to accuracy $\varepsilon$ with $O(N\log(1/\varepsilon))$ count and $\tilde O\left(N\log(1/\varepsilon)\frac{\log(n_{\text{anc}})}{n_{\text{anc}}}\right)$ depth of Clifford$+T$ gates, where $\tilde O$ suppresses the doubly logarithmic factors of $n_{\text{anc}}$. 
\end{theorem}

Theorem.~\ref{th:sp_main} achieves linear scaling of Clifford$+T$ count with respect to $N$, and this is applied for arbitrary space complexity. 
When $n_{\text{anc}}=O(n)$, the circuit depth is lower than the best-known result of $O\left(N\frac{\log(N/\varepsilon)}{n}\right)$.
Moreover, compared to~\cite{Sun.21} which also study the space-time trade-off of state preparation, our method improves the circuit depth scaling for a factor of $\tilde O(n_{\text{anc}}/\log n_{\text{anc}})$. Summary of some representative state preparation protocols are provided in Table.~1 and Table.~2.

The main idea of our construction is as follows (see also Fig.~1). For $n_{\text{anc}}=O(n)$, we construct the quantum state with a set of uniformly controlled rotations (UCR) with the method in~\cite{Mottonen.05}. Instead of decomposing each UCR with identical accuracy, we distribute the decomposition error in an optimized way. UCR with $m$ controlled qubits, denoted as $m$-UCR, should be decomposed into $2^m$ number of $m$-qubit controlled rotations. When performing Clifford$+T$ decomposition, to reduce the total circuit complexity, we allow larger decomposition accuracy when $m$ becomes larger. 

For $n_{\text{anc}}=O(N)$, we improve the Clifford$+T$ decomposition of the method in~\cite{Zhang.22} in a similar way. In both cases, the gate count scaling $O(N\log(1/\varepsilon))$ is achieved. For arbitrary ancillary qubit number between two extreme cases, we provide a scheme combing two protocols together, which allows space-time trade-off. Details of our state preparation scheme and the corresponding complexity analysis are provided in Methods.  
We also note that our protocol for few qubit case can be combined with the depth-optimal scheme in~\cite{Gui.23}. The circuit depth can then be improved to $O(N\log(1/\varepsilon)n_{\text{anc}}/\log(n_{\text{anc}}))$, at the cost of higher gate count.

We note that when the quantum state is sparse, the circuit complexity will be significantly lower. The construction of sparse state preparation is useful for sparse block-encoding. Details about sparse state preparation and sparse matrix block-encoding are provided in Supplementary Discussion 2. 

$\\$
\noindent\textbf{Other useful subroutines}

\begin{table}[h]
\caption{Summary of the Clifford$+T$ circuit complexities of the operations serving as subroutines in this work. Suppose the subroutine has $n_{\text{dat}}$ data qubit, in the last column, the first row corresponds to the circuit depth when there are $n_{\text{anc}}=O(n_{\text{dat}})$ ancillary qubit; the second row corresponds to the circuit depth without qubit number restriction. $\tilde{\tilde{O}}$ suppresses the doubly logarithmic factors with respect to $n$ and $s$. SP, SOPS, SBM and SSP correspond to state preparation, select oracle for Pauli strings, sparse Boolean memory and sparse state preparation (Supplementary Discussion 2), respectively.
\label{tab:0}}
\begin{ruledtabular}\footnotesize
\begin{tabular}{ccccccc}
Operations &Count &Depth\\
\hline
\multicolumn{1}{c}{\multirow{2}{*}{SP (Thm.~\ref{th:sp_main})}} & \multicolumn{1}{c}{\multirow{2}{*}{$O(N\log(1/\varepsilon))$}}&$\tilde{\tilde{O}}\left(N\frac{\log(n)\log(1/\varepsilon)}{n}\right)$ \\
\multicolumn{1}{c}{}&\multicolumn{1}{c}{}& $\tilde{\tilde{O}}(n\log(1/\varepsilon))$   \\
\hline
\multicolumn{1}{c}{\multirow{2}{*}{SOPS (Thm.~\ref{th:sop})}} & \multicolumn{1}{c}{\multirow{2}{*}{$O(ML)$}}&$O\left(ML\frac{\log(m+L)}{m+L}\right)$ \\
\multicolumn{1}{c}{}&\multicolumn{1}{c}{}& $O\left(m+\log L\right)$   \\
\hline
\multicolumn{1}{c}{\multirow{2}{*}{SBM (Thm.~\ref{th:sbm})}} & \multicolumn{1}{c}{\multirow{2}{*}{$O(ns\tilde n)$}}&$O(\log(n)s\tilde n)$ \\
\multicolumn{1}{c}{}&\multicolumn{1}{c}{}& $O\left(\log(ns\tilde n)\right)$\\
\hline
\multicolumn{1}{c}{\multirow{2}{*}{SSP (Thm.~9)}} & \multicolumn{1}{c}{\multirow{2}{*}{$O(s(n\log s+\log(1/\varepsilon)))$}}&$\tilde{\tilde{O}}(s\log(s)\log(n/\varepsilon))$ \\
\multicolumn{1}{c}{}&\multicolumn{1}{c}{}& $\tilde{\tilde{O}}\left(\log(s)\log(1/\varepsilon)+\log n\right)$    
\end{tabular}
\end{ruledtabular}
\end{table}

\noindent Before discussing the construction of access models in Eq.~\eqref{eq:ofh} and Eq.~\eqref{eq:bl}, we introduce some other useful subroutines, including select oracle and quantum sparse Boolean memory. These operations may have applications individually in some other scenarios. 
For both operations, we obtain their space-time trade-off constructions, which have improved or comparable Clifford$+T$ complexities compared to the best-known realizations (see also Table.~3).

\begin{figure}[t]
    \centering
        \includegraphics[width=\columnwidth]{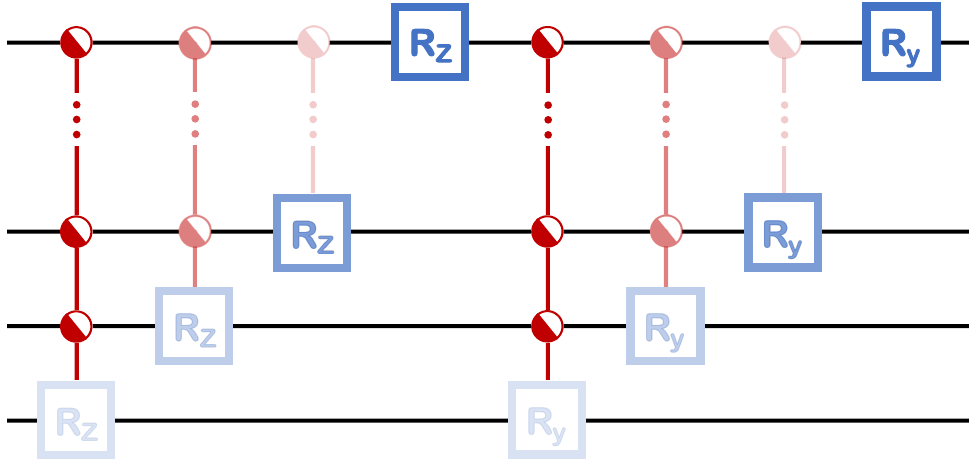}
       \caption{State preparation achieving $O(N)$ Clifford$+T$ count for few qubit case. The operation is decomposed into uniformly controlled Z- and Y-rotations, whose control and rotation parts are denoted with red and blue colors respectively. Each $m$-UCR is decomposed into $2^m$ multi-qubit controlled single-qubit rotations, and $m$ increases with the opacity of the control part (red). 
Each Z- or Y-rotation is decomposed into Clifford$+T$ gates, and the decomposition accuracy increases with the opacity of the rotation part (blue).}
       \label{fig:1}
\end{figure} 

$\\$
\noindent\textit{Select oracle for Pauli strings.} 
We consider a function of Pauli strings $H_x=\bigotimes_{l=1}^{L}H_{x,l}$, where $x\in\{0,1,\cdots,2^m-1\}$ and $H_{x,l}\in\{\pm I,\pm X,\pm Y,\pm Z\}$. We introduce two registers, the \textit{index register} contains $m$ qubits, and the \textit{word register} contains $L$ qubits. Select oracle for $H_x$ is defined as
\begin{align}\label{eq:slp}
\text{Select}(H_x)=\sum_{x=0}^{2^m-1}|x\rangle\langle x| \otimes H_x,
\end{align}
where $|x\rangle$ represents the computational basis of index register, and the unitary $H_x$ is applied at the word register. In other words, the state of index register controls the operations applied at the word register.

Several proposals of implementing Eq.~\eqref{eq:slp} has been introduced in the literature. For example, with $n_{\text{anc}}=m$ ancillary qubits, Ref~\cite{Childs.18} (Appendix G.4) proposed a method achieving $O(ML)$ circuit depth and gate count with $M=2^m$.  With $n_{\text{anc}}=O(ML)$ ancillary qubits, Eq.~\eqref{eq:slp} is a special form of the ``\textit{product unitary memory}'' in~\cite{Zhang.22}, which can be constructed with $O(\log(ML))$ depth and $O(ML)$ count of Clifford$+T$ gates. We provide an algorithm with tunable ancillary qubit number achieving the circuit complexity as follows.

\begin{theorem}\label{th:sop}
With $n_{\text{anc}}$ ancillary qubits where $\Omega(m+L)\leqslant n_{\text{anc}}\leqslant O(ML)$, Eq.~\eqref{eq:slp} 
can be realized with $O(ML)$ count and $O\left(ML\frac{\log n_{\text{anc}}}{n_{\text{anc}}}\right)$ depth of Clifford$+T$ gates.
\end{theorem}
Compared to the result in Ref~\cite{Childs.18}, our protocol reduces the circuit depth for a factor of $O\left(\frac{\log n_{\text{anc}}}{n_{\text{anc}}}\right)$ while maintaining the gate count scaling.
The proof of Theorem.~\ref{th:sop} and details of circuit constructions are provided in Methods. 

$\\$
\noindent\textit{Sparse Boolean memory.} 
We consider a sparse Boolean function $B:\{0,1\}^{n}\rightarrow\{0,1\}^{\tilde n}$, which has totally $s$ input digits $q$ satisfying $B(q)\neq0\cdots0$.  Given an $n$-qubit index register (denoted as idx) and a $\tilde n$-qubit register (denoted as wrd), we define the sparse Boolean memory $\text{Select}(B)$ as a unitary satisfying
\begin{align}\label{eq:sb}
\text{Select}(B)|q\rangle_{\text{idx}}|z\rangle_{\text{wrd}}=|q\rangle_{\text{idx}}|z\oplus B(q)\rangle_{\text{wrd}}.
\end{align}
We have the following result (see Methods for proof).
\begin{theorem}\label{th:sbm}
With $n_{\text{anc}}$ ancillary qubits where $\Omega(n)\leqslant n_{\text{anc}}\leqslant O(ns\tilde n)$, Select$(B)$ can be realized with $O(ns\tilde n)$ count and $O\left(ns\tilde n\frac{\log n_{\text{anc}}}{n_{\text{anc}}}\right)$ depth of Clifford$+T$ gates.
\end{theorem}

Different from SAIM, Eq.~\eqref{eq:sb} contains constant number of nonzero outputs. So its construction requires much less resource.

$\\$
\noindent\textbf{Construction of SAIM}

\noindent With all necessary tools ready, we now discuss the construction of the SAIM in Eq.~\eqref{eq:ofh}. We have the following result.

\begin{theorem}\label{th:saim}
Given $n_{\text{anc}}$ ancillary qubits where $\Omega(n)\leqslant n_{\text{anc}}\leqslant O(Nnds)$, $O_H$ can be constructed with $O(Nnds)$ count and $O\left(Nnds\frac{\log n_{\text{anc}}}{n_{\text{anc}}}\right)$ depth of Clifford$+T$ gates.

Given $n_{\text{anc}}$ ancillary qubits where $\Omega(n)\leqslant n_{\text{anc}}\leqslant O(Nns\log s)$, $O_F$ can be constructed with $O(Nns\log s)$ count and $O\left(Nns\log s\frac{\log n_{\text{anc}}}{n_{\text{anc}}}\right)$ depth of Clifford$+T$ gates.
\end{theorem}

\begin{proof}

 $O_H$ corresponds to a $2n$-index, $d$-word and $Ns$-sparse Boolean function. So the construction of $O_H$ follows directly from Theorem.~\ref{th:sbm}.

The construction of $O_F$ can be realized in three steps. We introduce an $n$-qubit ancillary register (denoted as anc). In the first step, we perform the following transformation 
\begin{align}
|x,k\rangle_{\text{idx}}|0\rangle_{\text{anc}}\rightarrow|x,k\rangle_{\text{idx}}|F(x,k)\rangle_{\text{anc}}.
\end{align}
According to Theorem.~\ref{th:sop}, this step can be constructed with 
$O(2^nns)$ count and $O(2^{n}ns\frac{\log n_{\text{anc}}}{n_{\text{anc}}})$ depth with $n_{\text{anc}}$ ancillary qubits. In the second step, we apply swap gates between the ancillary register and half of the index register which encodes $k$, i.e. 
\begin{align}
|x,k\rangle_{\text{idx}}|F(x,k)\rangle_{\text{anc}}\rightarrow|x,F(x,k)\rangle_{\text{idx}}|k\rangle_{\text{anc}}
\end{align}
This step can be realized with $O(n)$ count and $O(1)$ depth of Clifford$+T$ gates. In the final step, we perform the transformation 
\begin{align}
|x,F(x,k)\rangle_{\text{idx}}|k\rangle_{\text{anc}}\rightarrow|x,F(x,k)\rangle_{\text{idx}}|0\rangle_{\text{anc}}
\end{align}
which can be realized by a $2n$-index, $\lceil\log_2s\rceil$-word and $Ns$-sparse Boolean memory. According to Theorem.~\ref{th:sbm}, this step can be constructed with $O(Nns\log s)$ count and $O\left(Nns\log s\frac{\log n_{\text{anc}}}{n_{\text{anc}}}\right)$ depth of Clifford$+T$ gates. The total gate complexity is therefore the combination of three steps above.

\end{proof}

Compared to the circuit complexity lower bound obtained in Theorem.~\ref{th:lbsm}, our protocol has nearly optimal circuit complexities with respect to the matrix dimension up to a factor of $n$. As mentioned before, SAIM is a standard access model in many quantum algorithms, and the query complexity to SAIM has been studied extensively for various tasks. With Theorem.~\ref{th:saim}, one can directly obtain the natural circuit complexity of those algorithms. Further discussions are provided in the DISCUSSION section.

$\\$
\noindent\textbf{Construction of LCU-based Block-encoding}

\noindent The construction of LCU-based block-encoding can be realized with quantum state preparation and select oracle~\cite{Low.17}. We define $\bm{\alpha}=[\alpha_1,\cdots,\alpha_{P}]$ and $|\bm{\alpha}\rangle=\sum_{p=1}^P\sqrt{\alpha_p}|p\rangle$. Let $G_{|\bm{\alpha}\rangle}$ be the state preparation unitary for $|\bm{\alpha}\rangle$, and we define $\mathbb{G}\equiv G_{|\bm{\alpha}\rangle}\otimes\mathbb{I}_{2^n}$. We then define a Select oracle corresponding to Eq.~\eqref{eq:lcugen} as
$\text{Select}(u_p)=\sum_{p=0}^{P-1}|p\rangle\langle p|\otimes u_p$. 
It can be verified that $\mathbb{G}^\dag\text{Select}(u_p)\mathbb{G}$ is a block-encoding of $H$ with normalization factor $\alpha=1$~\cite{Low.19}. The constructions of LCU-based block-encoding is then reduced to the quantum state preparation and Select$(u_p)$, both of which can be constructed with polynomial-size quantum circuits.

The exact circuit complexity of block-encoding depends on the specific form of $u_p$. We take the LCU for Pauli strings (Eq.~\eqref{eq:lcu}) as an example. Based on our improved quantum state preparation (Theorem.~\ref{th:sp_main}) and Select oracle for Pauli strings (Theorem.~\ref{th:sop}), we have the following result, where $(n_{\text{anc}},\varepsilon)$-block-encoding is the abbreviation of $(1,n_{\text{anc}},\varepsilon)$-block-encoding (see Methods section for proof).

\begin{theorem}\label{th:bl1}
With $n_{\text{anc}}$ ancillary qubits where $\Omega(\log_2P)\leqslant n_{\text{anc}}\leqslant O(NP)$, the $(n_\text{anc},\varepsilon)$-block-encoding of $H$ defined in Eq.~\eqref{eq:lcu} can be constructed with  $O\left(P(n+\log(1/\varepsilon))\right)$ count and $\tilde O\left(Pn\log(1/\varepsilon)\frac{\log n_\text{anc}}{n_\text{anc}}\right)$ depth of Clifford$+T$ gates, where $\tilde O$ suppresses the doubly logarithmic factors of $n_{\text{anc}}$. 
\end{theorem}

The block-encoding of LCU can be constructed with polylogarithmic circuit complexity with respect to the data dimension, as oppose to the SAIM requiring polynomial gate count. Therefore, for structured classical data in the form of Eq.~\eqref{eq:lcugen} exponential quantum advantage can be expected. In below, we provide further discussions about by our results.

\section{Discussion}\label{sec:dis}

As demonstrated in Theorem.~\ref{th:lbsm}, a general SAIM can not be implemented with $O(\text{Poly}(n))$ size quantum circuit. In the language of complexity class, this implies that BQP$^{\text{SAIM}}\neq$BQP, where SAIM represent the quantum oracles in the form of Eq.~\eqref{eq:ofh}.  In other words, if problem $A$ can be solve with polynomial number of queries to the SAIM, $A$ is not necessarily solvable with polynomial-size quantum circuits. In fact, it is reasonable to conjecture that BQP$^{\text{SAIM}}\neq$PSPACE when considering the scaling with $n$. The reason is that for a general matrix with $2^n$ dimension, storing all its element requires exponentially large space, and this is true even for sparse matrix. The same argument applies to the block-encoding of sparse matrices as well.
 
 This argument is consistent with the results about classical dequantization algorithms~\cite{Tang.19,Chia.20b}, which demonstrate that sub-linear classical runtime can be achieved for tasks such as recommendation systems and solving linear systems.
 Note that these algorithms assumes a classical oracle similar to SAIM.
 
On the other hand, our study on sparse matrix encoding still has its great value. First of all, it is rare to have structured classical data that can be encoded with logarithmic complexity. In many cases, sparse matrix is the most compact representation for classical data of interest. Second, with SAIM or block-encoding,  polynomial quantum speedup with respect to the matrix dimension $N$ is still possible. Our constructions are nearly optimal, and can be used to estimate the concrete Clifford$+T$ complexities of many quantum algorithms of practical interest. Finally, techniques developed here may serve as a subroutine for encoding a larger matrix with special structures, with which the with which exponential quantum advantage may be possible.

An open question is how to determine whether a given matrix is efficiently block-encodable. This problem can be considere as a generalization of the unitary complexity problem~\cite{Nielsen.05,Bu.21,Eisert.21,Li.22}, which is important due to the broad applications of block-encoding~\cite{Martyn.21}.
According to Theorem.~\ref{th:bl1}, LCU for efficient unitaries [Eq.~\eqref{eq:lcugen}] is a sufficient condition of efficient block-encoding. Due to the generality and simplicity of LCU, it is reasonable to conjecture that the decomposition of a matrix in the form of Eq.~\eqref{eq:lcugen} has close relation to the efficiency of its block-encoding. The block encoding of $H$ is challenging when it can not be well approximated by Eq.~\eqref{eq:lcugen} with $P=O(\text{Poly}(n))$.

In conclusion, we have studied the circuit complexities of typical quantum access models, such as SAIM and block-encoding. We show that the circuit complexity lower bound for encoding sparse matrix is polynomial with respect to the matrix dimension. We provide nearly-optimal construction protocols to achieve the lower bound. For LCU-based block-encoding, we develop a construction protocol based on the improved implementation of quantum state preparation and select oracle for Pauli strings. Our protocols are based on Clifford$+T$ gates and allow tunable ancillary qubit number. We expect that our results are useful for processing classical data with quantum devices~\cite{Biamonte.17,Liu.22,Liu.23}. Future works may include the study of the circuit complexity lower bound for block-encoding, and how to further improve our protocols to achieve the lower bounds. Another interesting topic is about the power of quantum circuits with global quantum channels. For example, if the feedback controls are dependent on the measurement outcomes of many measurements. In this case, the elementary operations may no longer be described by local operations, and the computation power of the circuit is expected to be enhanced. In the direction of applications, it is interesting to find practical classical problems, whose data structure are able to be represented in the form LCU. In those scenarios, exponential quantum advantage can be expected. 

\section{Methods}

\noindent\textbf{Quantum state preparation}

\noindent We first consider the preparation with $n$ ancillary qubits. There are some state preparation protocol with optimal single- and two-qubit gate count, such as Ref~\cite{Mottonen.05}. However, with direct Clifford$+T$ decomposition, the gate complexity becomes suboptimal. We achieve gate count and circuit depth linear to the state dimension with an optimized Clifford$+T$ decomposition. The result is as follows.

\begin{lemma}\label{lm:sp_few}
With $n$ ancillary qubits, an arbitrary quantum state can be prepared to precision $\varepsilon$ with $O(N\log(1/\varepsilon))$  depth and $O(N\log(1/\varepsilon))$ count of Clifford $+T$ gates.
\end{lemma}

\begin{proof}
According to~\cite{Mottonen.05}, with single- and two-qubit gates, an arbitrary quantum state $|\psi_{\text{targ}}\rangle$ can be expressed as 
\begin{align}
|\psi_{\text{targ}}\rangle=\left(\prod_{j=1}^{n}F^{y}_j\right)\left(\prod_{j=1}^{n}F^{z}_j\right)|0\cdots0\rangle,
\end{align}
where $F^{z}_j$ and $F^{y}_j$ are uniformly controlled Z- and Y-rotations
\begin{subequations}\label{eq:F1}
\begin{align}
F^{z}_j=\sum_{k=0}^{2^{j-1}-1}|k\rangle\langle k|\otimes R_z(\alpha_{j,k}^z)\otimes\mathbb{I}_{2^{n-j}},\\
F^{y}_j=\sum_{k=0}^{2^{j-1}-1}|k\rangle\langle k|\otimes R_y(\alpha_{j,k}^y)\otimes\mathbb{I}_{2^{n-j}},
\end{align}
\end{subequations}
 with single qubit rotation gates $R_y(\theta)=e^{-i\theta\sigma_y/2}$, $R_z(\theta)=e^{-i\theta\sigma_z/2}$. Here $\alpha^y_{j,k}\in\mathbb{R}$ and $\alpha^z_{j,k}\in\mathbb{R}$ are some rotation angles, the exact values of which are not important for our analysis. 
 
Single-qubit rotations can be approximated with Clifford$+T$ gates. According to Ref~\cite{Selinger.12}, unitary $\widetilde u_z$ satisfying $\left\|\widetilde u_z-R_z(\alpha_{j,k}^z/2)\right\|\leqslant \varepsilon_j/2$
 can be constructed with $O(\log(1/\varepsilon_j))$ single-qubit Clifford$+T$ gates without ancilla. 
 Accordingly, we can implement single-qubit-controlled-$R_z(\alpha_{j,k}^z,\varepsilon_j)$, such that 
 \begin{align}
 \left\|R_z(\alpha_{j,k}^z,\varepsilon_j)-R_z(\alpha_{j,k}^z)\right\|\leqslant\varepsilon_j
 \end{align}
  with the following circuit. 
  
  \begin{center}
  \begin{quantikz}
&\qw&\ctrl{1}&\qw&\ctrl{1}&\qw\\
&\gate{\widetilde u_z}&\targ{}&\gate{\widetilde u_z^\dag}&\targ{}&\qw
\end{quantikz}
  \end{center}

\noindent Note that $\widetilde u_z^\dag$ can be realized by the inverse conjugation of the Clifford$+T$ gate sequence of $\widetilde u_z$. Similar argument is also applied for $R_y(\alpha^y_{j,k})$. Then, according to Lemma.~\ref{lm:sl1} as will be introduced in the next section, one can construct the following unitaries
\begin{align}
\widetilde F_{y}^j=\sum_{k=0}^{2^{j-1}-1}|k\rangle\langle k|\otimes \widetilde R_y(\alpha_{j,k}^y,\varepsilon_j)\otimes\mathbb{I}_{2^{n-j}}\\
\widetilde F_{z}^j=\sum_{k=0}^{2^{j-1}-1}|k\rangle\langle k|\otimes \widetilde R_z(\alpha_{j,k}^z,\varepsilon_j)\otimes\mathbb{I}_{2^{n-j}}
\end{align}
with $j$ ancillary qubits, $O(2^j\log(1/\varepsilon_j))$ depth and $O(2^j\log(1/\varepsilon_j))$ count of Clifford$+T$ gates. We therefore approximate the target state with the following
\begin{align}
|\widetilde\psi_{\text{targ}}\rangle=\left(\prod_{j=1}^{n}\widetilde{F}^{z}_j\right)\left(\prod_{j=1}^{n}\widetilde{F}^{y}_j\right)|0\cdots0\rangle.
\end{align}
In below, we first bound the distance between $|\psi_{\text{targ}}\rangle$ and $|\widetilde\psi_{\text{targ}}\rangle$. It can be verified that  $\left\|\widetilde F^{y}_j|\psi\rangle-F^{y}_j|\psi\rangle\right\|\leqslant \varepsilon_j$ and $\left\|\widetilde F^{z}_j|\psi\rangle-F^{z}_j|\psi\rangle\right\|\leqslant\varepsilon_j$ for any quantum state $|\psi\rangle$. In other words, we have $\left\|\widetilde F^{y}_j-F^{y}_j\right\|\leqslant\varepsilon_j$ and $\left\|\widetilde F^{z}_j-F^{z}_j\right\|\leqslant\varepsilon_j$.
Therefore,
\begin{align}
&\left\|\left(\prod_{j=1}^{n}\widetilde F^{y}_j-\prod_{j=1}^{n}F^{y}_j\right) \right\|\notag\\
\leqslant&
\left\|\widetilde F^{y}_{n}\left(\prod_{j=1}^{n-1}\widetilde F^{y}_j-\prod_{j=1}^{n-1}F^{y}_j\right) \right\|
+
\left\|\left(\widetilde F^{y}_{n}-F^{y}_{n}\right)\prod_{j=1}^{n-1}F^{y}_j\right\|\notag\\
\leqslant&\left\|\left(\prod_{j=1}^{n-1}\widetilde F^{y}_j-\prod_{j=1}^{n-1}F^{y}_j\right) \right\|+ \varepsilon_{n}\notag\\
&\cdots\notag\\
\leqslant&\sum_{j=1}^n\varepsilon_{j}.
\end{align}
In a similar way, we can obtain $\left\|\left(\prod_{j=1}^{n}\widetilde F^{z}_j-\prod_{j=1}^{n}F^{z}_j\right) \right\|\leqslant\sum_{j=1}^n\varepsilon_{j}$. So we have  
\begin{align}\label{eq:targ_err0}
&\left\||\psi_{\text{targ}}\rangle-|\widetilde{\psi}_{\text{targ}}\rangle\right\|\notag\\
\leqslant&\left\|\left(\prod_{j=1}^{n}\widetilde F^{y}_j\prod_{j=1}^{n}\widetilde F^{z}_j-\prod_{j=1}^{n}F^{y}_j\prod_{j=1}^{n}F^{z}_j\right) \right\|\notag\\
\leqslant&\left\|\left(\prod_{j=1}^{n}\widetilde F^{z}_j-\prod_{j=1}^{n}F^{z}_j\right) \right\|+\left\|\left(\prod_{j=1}^{n}\widetilde F^{y}_j-\prod_{j=1}^{n}F^{y}_j\right) \right\|\notag\\
\leqslant&2\sum_{j=1}^n\varepsilon_{j}.
\end{align}
According to Eq.~\eqref{eq:targ_err0}, to control the total error rate to a constant value i.e. $\||\psi_{\text{targ}}\rangle-|\widetilde{\psi}_{\text{targ}}\rangle\|\leqslant\varepsilon$, it suffice to set $\varepsilon_j=\varepsilon/2^{n-j+1}$.
Because each $\widetilde F^y_j$ or $\widetilde F^z_j$ require $O(2^j\log(1/\varepsilon_j))$ gate count and circuit depth, the total gate count is 
\begin{align}
C&=\sum_{j=0}^{n-1}O(2^j\log(1/\varepsilon_j))\notag\\
&=\sum_{j=0}^{n-1}O(2^j\log(2^{n-j}/\varepsilon))\notag\\
&=O(N\log(1/\varepsilon)).
\end{align} 
Similarly, the total circuit depth is 
\begin{align}
D&=\sum_{j=0}^{n-1}O(2^j\log(1/\varepsilon_j))=O(N\log(1/\varepsilon)).
\end{align}
\end{proof}

We then consider the quantum state preparation with exponential ancillary qubits. Our protocol follows the same idea in~\cite{Zhang.22} with improvement.

\begin{lemma}\label{lm:sp_many}
Arbitrary $n$-qubit quantum state can be prepared with $O(N)$ ancillary qubits, $O(n\log(n/\varepsilon))$ depth and $O(N\log(1/\varepsilon))$ count of Clifford$+T$ gates.
\end{lemma}

\begin{proof}
Our construction is based on the protocol in~\cite{Zhang.22} with revision and improved Clifford+$T$ decomposition. 

\textit{General procedure}. The hardware layout of our method contains a binary tree of qubits with $n+1$ layers, which is denoted as $H$. The $l$th (with $0\leqslant l\leqslant n$) layer of $H$ is denoted as $H_l$. For $1\leqslant l\leqslant n$, $H_l$ connects to another binary tree of qubits, denoted as $V_l$. The root of the tree $V_l$ serves as the $l$th data qubit, and we denote it as $\text{d}_l$ here. 

Our protocol for preparing target state $|\psi_{\text{targ}}\rangle
=\sum_{k=0}^{2^n-1}\alpha_k|k\rangle_{\text{d}}
$ works as follows. We initialize the root of $H$ as $|1\rangle_{H_1}$ while all other qubits are at state $|0\rangle$.   In the first stage, $H$ is prepared at the quantum state (qubits at state $|0\rangle$ are not shown)
\begin{align}\label{eq:stg1}
|1\rangle_{H_1}\rightarrow\sum_{k=0}^{2^n-1}\alpha_k|\varphi_k\rangle_{H}.
\end{align}
Here, $|\varphi_k\rangle$ is one of the computational basis of $H$ to be defined later. In the second stage,  the data qubits are transferred to the $n$-qubit computational basis $|k\rangle_{\text{d}}$ conditioned on $|\varphi_k\rangle$, i.e.
\begin{align}\label{eq:stg2}
\sum_{k=0}^{2^n-1}\alpha_k|\varphi_k\rangle_{H}\rightarrow\sum_{k=0}^{2^n-1}\alpha_k|\varphi_k\rangle_{H}|k\rangle_{\text{d}}.
\end{align}
Finally, the binary tree $H$ is uncomputed 
\begin{align}\label{eq:stg3}
\sum_{k=0}^{2^n-1}\alpha_k|\varphi_k\rangle_{H}|k\rangle_{\text{d}}\rightarrow\sum_{k=0}^{2^n-1}\alpha_k|0\rangle_{H}|k\rangle_{\text{d}}.
\end{align}
The target state is then obtained after tracing out $H$. The readers are refereed to~\cite{Zhang.22} for more details.
Transformations in Eq.~\eqref{eq:stg2} and Eq.~\eqref{eq:stg3} can be ideally realized using Clifford circuit with $O(n)$ depth and $O(2^n)$ gate count. On the other hand, the first stage for obtaining Eq.~\eqref{eq:stg1} contains rotation that has to be approximated with $T$ gates and hence more complicated. So we focus on Eq.~\eqref{eq:stg1} in below. 

\textit{Realization of Eq.~\eqref{eq:stg1}.} We will first show how Eq.~\eqref{eq:stg1} can be realized with single-qubit and CNOT gates with a method slightly different from~\cite{Zhang.22}, and then introduce its Clifford $+T$ decomposition.

We define $\alpha_{n,k}\equiv a_{k}$ and $\alpha_{L,k}=\text{arg}(\alpha_{L+1,2k})\sqrt{|\alpha_{L+1,2k}|^2+|\alpha_{L+1,2k+1}|^2}$
for all $0\leqslant l\leqslant n-1$. Note that we can assume $\text{arg}(\alpha_0)=0$ without loss of generality. For $0\leqslant L\leqslant n$, we define
\begin{align}\label{eq:stpL}
|\Psi_{L}\rangle=&\sum_{k=0}^{2^{L}-1}\alpha_{L,k}\bigotimes_{l=0}^{L}|(k,l)\rangle'_{H_{l}}.
\end{align}
The realization of Eq.~\eqref{eq:stg1} contains $n$ steps, with the $L$th step corresponds to $|\Psi_{L-1}\rangle\rightarrow|\Psi_{L}\rangle$.

In Eq.~\eqref{eq:stpL}, we have defined $(0,0)\equiv0$, and $(k,l)\equiv k_nk_{n-1}\cdots k_{n-l+1}$ for $l\geqslant 1$; $|(k,l)\rangle'\equiv|0\rangle^{\otimes (k,l)}|1\rangle|0\rangle^{\otimes 2^l-(k,l)-1}$; $H_l$ represents the $l$th layer of $H$. 
 Eq.~\eqref{eq:stg1} and Eq.~\eqref{eq:stpL} have the correspondence $|\varphi_k\rangle=\bigotimes_{l=0}^{L}|(k,l)\rangle'_{H_l}$ and $\sum_{k=0}^{2^n-1}\alpha_k|\varphi_k\rangle_{H}=|\psi_{n}\rangle$. So $|\Psi_n\rangle$ is the target state of the stage 1 introduced in Eq.~\eqref{eq:stg1}.

We then introduce the realization of  $|\Psi_{L-1}\rangle\rightarrow|\Psi_{L}\rangle$.
 We define single qubit rotation $r_y(\theta)=\begin{pmatrix}\cos\theta&\sin\theta\\-\sin\theta&\cos\theta\end{pmatrix}$ and $r_z(\phi)=\begin{pmatrix}e^{-i\phi}&0\\0&e^{i\phi}\end{pmatrix} 
$, and a three-qubit controlled operation as follows.

\begin{widetext}
\begin{center}
\begin{quantikz}
\lstick{$a$}&\ctrl{1}&\qw&\ctrl{1}&\qw&\qw&\ctrl{1}&\qw&\ctrl{1}&\ctrl{1}&\qw\\
\lstick{$w(\theta;\phi;a;b;c)\quad=\quad b$}&\targ{}&\gate{r_y(\theta/2)}&\targ{}&\gate{r_y(\theta/2)^\dag}&\gate{r_z(\phi/2)^\dag}&\targ{}&\gate{r_z(\phi/2)}&\targ{}&\octrl{1}&\qw\\
\lstick{$c$}&\qw&\qw&\qw&\qw&\qw&\qw&\qw&\qw&\targ{}&\qw&\\
\end{quantikz}

\end{center}
\end{widetext}

$a,b,c$ are labels of the corresponding qubits. Let $\theta_{l,j}\equiv \arccos (b_{l,2j}/b_{l-1,j})$ and $\phi_{l,j}=\phi_{l,2j+1}-\phi_{l,2j}$, at the $L$th step ($1\leqslant L\leqslant n$), we implement the parallel rotation
\begin{align}
W_L=\prod_{j=0}^{2^{L-1}-1}w(\theta_{L,j};\phi_{L,j};H_{L-1,j};H_{L,2j};H_{L,2j+1})
\end{align}
which costs $O(1)$ depth and $O(2^L)$ count of single-qubit and CNOT gates.  It can be verified that 
\begin{align}
W_L|\Psi_{L-1}\rangle=|\Psi_{L}\rangle.
\end{align}
The total single-qubit $+$ CNOT depth and gate count are $O(n)$ and $O(2^n)$ respectively.

\textit{Clifford$+T$ decomposition}. $W_L$ are assumed to be constructed with single- and two-qubit gates. In below, we discuss how to decompose it with Clifford +$T$  gates with high accuracy. According to Ref~\cite{Selinger.12}, one can always construct a unitaries $\tilde r_y(\theta;\varepsilon)$, $\tilde r_z(\phi;\varepsilon)$, with $O(\log(1/\varepsilon))$ depth of gates in $\{H,S,T\}$, which satisfies
\begin{align}
\|\tilde r_y(\theta;\varepsilon)- r_y(\theta)\|\leqslant\varepsilon, \quad \|\tilde r_z(\phi;\varepsilon)- r_z(\phi)\|\leqslant\varepsilon.
\end{align}
Accordingly, we define $\widetilde w(\theta;\phi;\varepsilon;a;b;c)$ as the following transformation
\begin{widetext}
\begin{center}
\begin{quantikz}
\lstick{$a$}&\ctrl{1}&\qw&\ctrl{1}&\qw&\qw&\ctrl{1}&\qw&\ctrl{1}&\ctrl{1}&\qw\\
\lstick{$\widetilde w(\theta;\phi;a;b;c)\quad=\quad b$}&\targ{}&\gate{\tilde r_y(\theta/2;\varepsilon)}&\targ{}&\gate{\tilde r_y(\theta/2;\varepsilon)^\dag}&\gate{\tilde r_z(\phi/2;\varepsilon)^\dag}&\targ{}&\gate{\tilde r_z(\phi/2;\varepsilon)}&\targ{}&\octrl{1}&\qw\\
\lstick{$c$}&\qw&\qw&\qw&\qw&\qw&\qw&\qw&\qw&\targ{}&\qw&\\
\end{quantikz}

\end{center}
\end{widetext}

We have
\begin{align}\label{eq:tw}
&\widetilde w(\theta;\phi;\varepsilon;a;b;c)(a|0\rangle_a|0\rangle_{b}|0\rangle_{c}+b|1\rangle_a|0\rangle_{b}|0\rangle_{c})\notag\\
=&a|0\rangle_a|0\rangle_{b}|0\rangle_{c}+b_1(\varepsilon)|1\rangle|10\rangle+b_2(\varepsilon)|01\rangle,
\end{align}
for some $\sqrt{|b_1(\varepsilon)-b_1(0)|^2+|b_2(\varepsilon)-b_2(0)|^2}\leqslant|b|\varepsilon$.
We then define 
\begin{align}
\widetilde W_L(\varepsilon)=\prod_{j=0}^{2^{L-1}-1}\widetilde w(\theta_{L,j};\phi_{L,j};\varepsilon;H_{L-1,j};H_{L,2j};H_{L,2j+1}),
\end{align}
which is used to approximate $W_L$.
From Eq.~\eqref{eq:tw}, it can be verified that
$\left\|\widetilde W_L(\varepsilon)|\Psi_{L-1}\rangle-W_L|\Psi_{L-1}\rangle\right\|\leqslant\varepsilon$. We set the accuracy at the $L$th layer as $\varepsilon_L$, and define 
\begin{align}
|\widetilde{\Psi}_{0}\rangle=|\Psi_{0}\rangle,\quad |\widetilde{\Psi}_{L}\rangle=\widetilde W_L(\varepsilon_L)|\widetilde{\Psi}_{L-1}\rangle.
\end{align}
We have
\begin{align}
&\left\||\widetilde\Psi_{L}\rangle-|\Psi_{L}\rangle\right\|\notag\\
=
&\left\|\widetilde W_L(\varepsilon_L)|\widetilde\Psi_{L-1}\rangle-W_L|\Psi_{L-1}\rangle\right\|\notag\\
\leqslant
&\left\|\widetilde W_L(\varepsilon_L)|\widetilde\Psi_{L-1}\rangle-\widetilde W_L(\varepsilon_L)|\Psi_{L-1}\rangle\right\|\notag\\
+&\left\|\widetilde W_L(\varepsilon_L)|\Psi_{L-1}\rangle-W_L|\Psi_{L-1}\rangle\right\| \notag\\
\leqslant 
&\left\|\widetilde W_L(\varepsilon_L)|\widetilde\Psi_{L-1}\rangle-\widetilde W_L(\varepsilon_L)|\Psi_{L-1}\rangle\right\|+\varepsilon_L\notag\\
=&\left\||\widetilde\Psi_{L-1}\rangle-|\Psi_{L-1}\rangle\right\|+\varepsilon_L.\label{eq:WL}
\end{align}
By applying the inequality above iteratively from $L=1$ to $L=n$, we have 
\begin{align}\label{eq:targ_err1}
\left\||\widetilde\Psi_n\rangle-|\Psi_n\rangle\right\|\leqslant\sum_{L=1}^{n}\varepsilon_{L}.
\end{align}
According to Eq.~\eqref{eq:targ_err1}, to control the total error rate to a constant value, it suffices to set $\varepsilon_{L}=K\varepsilon/(n-L+1)^2$ for some constant $K$. This is the key step of our improved construction.

\textit{Circuit complexity.}  Each $\widetilde W_L$ can be realized with $O(2^L\log(1/\varepsilon_L))$ count and $O(\log(1/\varepsilon_L))$ depth of Clifford$+T$ gates. Therefore, the total gate count at stage 1 (Eq.~\eqref{eq:stg1}) is
\begin{align}\label{eq:stg1C}
C&=O\left(\sum_{L=1}^{n}2^{L}\log(1/\varepsilon_{L})\right)\\
&=O\left(\sum_{L=1}^{n}2^{L}\log((n-L+1)^2/\varepsilon)\right)\notag\\
&=O\left(2^{n+1}\sum_{m=1}^{n}\frac{\log(m)}{2^{m}}\right)+O\left(2^{n}\log(1/\varepsilon)\right)\notag\\
&=O\left(2^{n}\right)+O\left(2^{n}\log(1/\varepsilon)\right)\notag\\
&=O\left(N\log(1/\varepsilon)\right).\notag
\end{align}
The total circuit depth at stage 1 is 
\begin{align}\label{eq:stg1D}
D&=O\left(\sum_{L=1}^{n}\log(1/\varepsilon_{L})\right)\notag\\
&=O\left(\sum_{L=1}^{n}\log((n-L+1)^2/\varepsilon)\right)\notag\\
&=O\left(\sum_{m=1}^{n}\log(m^2)\right)+O\left(n\log(1/\varepsilon)\right)\notag\\
&=O\left(\log(n!)\right)+O\left(n\log(1/\varepsilon)\right)\notag\\
&=O\left(n\log(n/\varepsilon)\right).
\end{align}
Recall that Eq.~\eqref{eq:stg2} and Eq.~\eqref{eq:stg3} has $O(N)$ count and $O(n)$ depth of Clifford$+T$ gates. So the total gate count and circuit depth are $O\left(N\log(1/\varepsilon)\right)$ and $O\left(n\log(n/\varepsilon)\right)$ respectively.

\end{proof}

We also cares about the controlled quantum state preparation. In our preparation scheme, the initial state is $|1\rangle_{H_1}$, i.e. the root of $H$ is set as $|1\rangle$. If we set $H_1$ as $|0\rangle_{H_1}$ instead, it can be verified that the output state is $|0\cdots0\rangle_{\text{d}}$. Therefore, to implement controlled state preparation, one can simply replace the root qubit $H_1$ by the controlled qubit, and the circuit complexity remains unchanged. In other words, we have the following result.

\begin{lemma}\label{lm:sp_many_ctrl}
Arbitrary single-qubit-controlled $n$-qubit state preparation unitary can be constructed with $O(N)$ ancillary qubits, $O(n\log(n/\varepsilon))$ depth and $O(N\log(1/\varepsilon))$ count of Clifford$+T$ gates.
\end{lemma}

Based on Lemma.~\ref{lm:sp_few}, Lemma.~\ref{lm:sp_many} and Lemma.~\ref{lm:sp_many_ctrl}, We have the following result for intermediate number of ancillary qubits. Note that Theorem.~\ref{th:sp_main} in the main text follows directly from Theorem.~\ref{th:stsp}.

\begin{theorem}[space-time tradeoff QSP]\label{th:stsp}
With $n_{\text{anc}}$ ancillary qubits where $\Omega(n)\leqslant n_{\text{anc}}\leqslant O(2^n)$, state preparation and controlled state preparation of an arbitrary $n$-qubit quantum state can be realized with precision $\varepsilon$ with $O(N\log(1/\varepsilon))$ count and $O\left(N\frac{\log(n_{\text{anc}})\log(\log(n_{\text{anc}})/\varepsilon)}{n_{\text{anc}}}\right)$ depth of Clifford$+T$ gates.
\end{theorem}

\begin{proof}
We separate all data qubits into two registers. Register $A$ contains the last $n_a=n-\lfloor\log_2m\rfloor$ data qubits, and register $B$ contains the first $n_b=\lfloor\log_2m\rfloor$ qubits for some $n\leqslant m\leqslant 2^n$. We define $N_a=2^{n_a}$ and $N_b=2^{n_b}$. The target state can be rewritten as
\begin{align}
|\psi_{\text{targ}}\rangle=\sum_{k=0}^{N_a-1}\beta_k|k\rangle_A|\phi_k\rangle_B
\end{align}
for some normalized $\beta_k$, and normalized quantum states $|\phi_k\rangle$. We define $|\psi_a\rangle_A=\sum_{k=0}^{N_a-1}\beta_k|k\rangle_A$. 

In the first step, we prepare register $A$ to a quantum state
\begin{align}
|\widetilde{\psi}_a\rangle_A=\sum_{k=0}^{N_a-1}\tilde{\beta}_k|k\rangle_A
\end{align}
which satisfies $\left\||\psi_a\rangle-|\widetilde{\psi}_a\rangle\right\|\leqslant\varepsilon/2$. According to Lemma.~\ref{lm:sp_few}, this step can be realized with $O(N_a\log(1/\varepsilon))=O(\frac{N}{m}\log(1/\varepsilon))$ count and depth of Clifford$+T$ circuit. 

In the second step, we implement
\begin{align}
&\text{Select}(\widetilde{G}_k)\sum_{k=0}^{N_a-1}\widetilde\beta_k|k\rangle_A|0^{n_b}\rangle_B=\sum_{k=0}^{N_a-1}\widetilde{\beta}_k|k\rangle_A|\widetilde{\phi}_k\rangle_B\notag\\
\equiv&|\widetilde{\psi}_{\text{targ}}\rangle
\end{align}
where $\widetilde{G}_k$ is a state preparation unitary satisfying $\widetilde{G}_k|0\rangle=|\widetilde\phi_k\rangle$ for some $\left\||\phi_k\rangle-|\widetilde\phi_k\rangle\right\|\leqslant \varepsilon/2$. It can be then verified that $\left\||\psi_\text{targ}\rangle-|\widetilde\psi_\text{targ}\rangle\right\|\leqslant\varepsilon$. 
According to Lemma.~\ref{lm:sp_many}, controlled-$\widetilde{G}_k$ such that $\widetilde{G}_k|0^{n_b}\rangle=|\widetilde\phi_k\rangle$ can be constructed with $O(m)$ ancillary qubits, $O(N_b\log(1/\varepsilon))$ count and $O(n_b\log(n_b/\varepsilon))=O(\log(m)\log(\log(m)/\varepsilon))$ depth of Clifford$+T$ gates. Then, according to Lemma.~\ref{lm:sl1}, with  $O(m)$ ancillary qubits, Select$(\widetilde{G}_k)$ can be constructed with  
\begin{align}
C=O(N_a\times N_b\log(1/\varepsilon))=O(N\log(1/\varepsilon))
\end{align}
gate count, and 
\begin{align}
D&=O\left(N_a\times\log(m)\log(\log(m)/\varepsilon)\right)\notag\\
&=O\left(N\frac{\log(m)\log(\log(m)/\varepsilon)}{m}\right)
\end{align}
 depth of Clifford$+T$ gates. By setting $n_{\text{anc}}=O(m)$ for some $n_{\text{anc}}\geqslant n$, we complete the proof.
\end{proof}

\noindent\textbf{Select oracle for general unitary functions.}
Suppose $x$ is an $m$-bit bitstring, and $U_x$ are general unitaries. We consider the unitary 
\begin{align}\label{eq:selectu}
\text{Select}(U_x)=\sum_{x=0}^{M-1}|x\rangle\langle x|\otimes U_x,
\end{align}
where $M=2^m$.
In below, we discuss how to construct Select$(U_x)$ based on the implementation of single-qubit-controlled-$U_x$, and the corresponding circuit complexity upper bound. We define $C_{\text{ctrl}}(U_x,r)$ and $D_{\text{ctrl}}(U_x,r)$ as the count and depth of Clifford$+T$ gates required to construct the controlled-$U_x$,  given $r$ ancillary qubits. The following result corresponds to the case with $m+r$ ancillary qubits.

\begin{lemma}[Appendix G.4 of~\cite{Childs.18}]\label{lm:sl1}
With $m+r$ ancillary qubits, $\rm{Select}$$(U_x)$ can be constructed with $O(MC_{\text{ctrl}}(U_x,r))$ count and $O(MD_{\text{ctrl}}(U_x,r))$ depth of Clifford$+T$ gates.
\end{lemma}

\begin{proof}
We introduce an ancillary register with $m$ qubits. 
We denote the $j$th qubit at the index register (encoding $|x\rangle$) and ancillary registers as $C_j$, $A_j$ respectively. We also denote $\bm{C}=[C_1,C_2,\cdots,C_m]$ and $\bm{A}=[A_0,A_1,A_2,\cdots,A_m]$. $A_0$ is initialized as $|1\rangle$ while all other ancillary qubits are initialized as $|0\rangle$.

Eq.~\eqref{eq:selectu} can be realized by querying Select$(\bm{C},\bm{A},m,0)$, which is defined recursively by Algorithm.~\ref{alg:select}. In Algorithm.~\ref{alg:select}, Toffoli($a,b;c$) is the Toffoli gate with qubit $a$ and $b$ as the controlled qubits and $c$ as the target qubit; C-$U_x(a)$ is the controlled-$U_x$ with qubit $a$ as controlled qubit and the corresponding word register as target qubits;   $\text{dim}(\bm{v})$ represent the dimension of the vector (for example, $\text{dim}(\bm{C})=n$); $v_{j}$ represents the $j$th element of $\bm{v}$ and $\bm{v}_{j:}=[v_j,v_{j+1},\cdots,v_{\text{dim}(\bm{v})}]$.

\begin{algorithm} [H]
\caption{ Select$(\bm{y},\bm{q},l,x)$
\label{alg:select}}
\begin{algorithmic}[1]
\STATE \textbf{if} $l\neq0$:   
\STATE \quad Toffoli($\overline{y_1},q_1;q_2$)
\STATE \quad Select($\bm{y_{2:}},\bm{q_{2:}}$,$l-1$,$x$)
\STATE \quad Toffoli($\overline{y_1},q_1;q_2$)
\STATE \quad Toffoli($y_1,q_1;q_2$)
\STATE \quad Select($\bm{y_{2:}},\bm{q_{2:}}$,$l-1$,$x+2^{l-1}$)
\STATE \quad Toffoli($y_1,q_1;q_2$)
\STATE \textbf{elseif} $l=0$:
\STATE \quad C-$U_x(q_1)$
\STATE \textbf{end if}
\end{algorithmic}
\end{algorithm}

In our implementation, the controlled-$U_x$ are queried for totally $M$ times with $x\in\{0,\cdots,m-1\}$ sequentially. Moreover, there are totally $O(M)$ Toffoli gates acting sequentially. Therefore, the total gate count and circuit depth are $O(MC_{\text{ctrl}}(U_x,r))$ and $O(MD_{\text{ctrl}}(U_x,r))$  respectively. 

\end{proof}

We note that Algorithm.~\ref{alg:select} can be further simplified by combining some concatenated gates~\cite{Childs.18}. But the asymptotic scaling here is optimal.

We then consider the construction of expoential  ancillary qubits. In Algorithm 4,5 of~\cite{Zhang.22}, based on the bucket-brigade architecture for quantum random access memory~\cite{Giovannetti.08,Hann.19,Hann.21}, it has been shown that any Select$(U_x)$ can be constructed by $4M-1$ ancillary qubits, $O(M)$ Clifford$+T$ gates arranged in $O(m)$ circuit depth, and queries to  all single-qubit-controlled-$U_x$ for $x\in\{0,\cdots,M-1\}$ in parallel. If each controlled-$U_x$ uses $r$ ancillary qubits, we require totally $M(4+r)-1$ ancillary qubits, because they are implemented in parallel. To sum up, we have the following result.

\begin{lemma}[many qubit Select oracle]\label{lm:sl2}

With $M(4+r)-1$ ancillary qubits, $\rm{Select}$$(U_x)$ can be constructed with $O(MC_{\text{ctrl}}(U_x,r))$ count and $O(m+D_{\text{ctrl}}(U_x,r))$ depth of Clifford$+T$ gates.
\end{lemma}

\noindent\textbf{Select oracle for general unitary functions.} 
In below, we give the proof of Theorem.~\ref{th:sop} about the construction of select oracles for Pauli strings defined in 
Eq.~\eqref{eq:slp}. Note that Eq.~\eqref{eq:slp} is a special case of Eq.~\eqref{eq:selectu} with $U_x\in\{\pm I,\pm X,\pm Y,\pm Z\}^{\otimes L}$.

\begin{proof}[proof of Theorem.~\ref{th:sop}.] 
Recall that Select oracle for Pauli strings corresponds to Eq.~\eqref{eq:selectu} with $U_x=H_x$, where $H_x=\bigotimes_{l=1}^{L}H_{x,l}$ and $H_{x,l}\in\{\pm I,\pm X,\pm Y,\pm Z\}$.

Given $L$ ancillary qubits, controlled-$H_x$ can be constructed with the following circuit. 
\begin{center}
\begin{quantikz}
\lstick{$c$}&\ctrl{6}&\qw&\qw&\ctrl{6}&\qw\\
\lstick{$a_1$}&\targ{}&\ctrl{1}&\qw&\targ{}&\qw\\
\lstick{$t_{1}$}&\qw&\gate{H_{x,1}}&\qw&\qw&\qw\\
\lstick{$a_2$}&\targ{}&\ctrl{1}&\qw&\targ{}&\qw\\
\lstick{$t_{2}$}&\qw&\gate{H_{x,2}}&\qw&\qw&\qw\\
\cdots&&\cdots&\cdots&&\cdots&\\
\lstick{$a_{L}$}&\targ{}&\ctrl{1}&\qw&\targ{}&\qw\\
\lstick{$t_{L}$}&\qw&\gate{ H_{x,L}}&\qw&\qw&\qw
\end{quantikz}
\end{center}
where control qubit is denoted as $c$, ancillary qubits, all initialized as $|0\rangle$, are denoted as $a_1, a_2,\cdots,a_{L}$ and target qubits are denoted as $t_1, t_2,\cdots,t_{L}$ respectively. Two of the $L$-Toffoli gates can be effectively constructed with $O(L)$ count and $O(\log L)$ depth of Clifford$+T$ gates. All controlled Pauli gates can be constructed with totally $O(L)$ count and $O(1)$ depth of Clifford$+T$ gates. In other words, we have $C_{\text{ctrl}}(H_x,L)=O(L)$ and $D_{\text{ctrl}}(H_x,L)=O(\log L)$.

Our protocol of constructing Select$(H_x)$ uses at least $\Omega(m+L)$ ancillary qubits. 
We divide the $m$-qubit index registers into two subregisters $A$ and $B$ with $m_a\geqslant\log_2(m+L)$ and  $m_b=m-m_a$ qubits respectively. Let $M_a=2^{m_a}, M_b=2^{m_b}$, $\text{Select}(H_{x})$ can be rewritten as
\begin{align}
 \text{Select}(H_{x}) &=\sum_{x_a=0}^{M_a-1}|x_a\rangle\langle x_a|\otimes V_{x_a}\\
V_{x_a} &=\sum_{x_b=0}^{M_b-1}|x_b\rangle\langle x_b|\otimes{H}_{x_{a}\oplus x_b}.
\end{align}
$x_a$ and $x_b$ are bit strings with $m_a$ and $m_b$ bits respectively, and $x\equiv x_a\oplus x_b$. According to Lemma.~\ref{lm:sl2}, $V_{x_a}$ can be constructed with $M_a(4+L)-1$ ancillary qubits, $O(M_aL)$ count and $O(m_a+\log L)$ depth of Clifford$+T$ gates. According to Lemma.~\ref{lm:sl1}, with totally $n_{\text{anc}}=M_a(4+L)-1+m_b$ ancillary qubits, the Clifford$+T$ gate count of Select$(H_x)$ is 
\begin{align}
C=O(M_bM_aL)=O(ML).
\end{align}
The Coifford$+T$ depth is  
\begin{align}
D&=O(M_b(m_a+\log L))\notag\\
&=O\left(M\frac{\log(M_aL)}{M_a}\right)\notag\\
&=O\left(M\frac{\log((n_{\text{anc}}+1)/4)}{(n_{\text{anc}}+1)/4L}\right)\notag\\
&=O\left(ML\frac{\log n_{\text{anc}}}{n_{\text{anc}}}\right),
\end{align}
which completes the proof.
\end{proof}

$\\$
\noindent\textbf{Details about LCU-based Block-encoding}
 
\noindent Without loss of generality, we assume that $m=\log_2P$ is an integer. We let $\tilde G_{|\bm{\alpha}\rangle}$ be a  state preparation unitary satisfying $\left\|\tilde G_{|\bm{\alpha}\rangle}|0^{m}\rangle_{\text{anc}}-G_{|\bm{\alpha}\rangle}|0^{m}\rangle_{\text{anc}}\right\|\leqslant\varepsilon/3$. Let $\tilde u_p$ be unitaries satisfying $\|\tilde u_p-u_p\|\leqslant\varepsilon/3$. 
We then define

\begin{align} 
\tilde{\mathbb{G}}&\equiv\tilde G_{|\bm{\alpha}\rangle}\otimes\mathbb{I}_{N},\\
 U&\equiv\mathbb{G}^\dag \text{Select}(u_p)\mathbb{G},\\ 
 \tilde U&\equiv\tilde{\mathbb{G}}^\dag \text{Select}(u_p)\tilde{\mathbb{G}},\label{eq:tg}
\end{align}
and $W\equiv\tilde U-U$. With a similar argument to Eq.~\eqref{eq:WL}, we have 
 
 \begin{align}\label{eq:wpsi0}
 \left\|W|\Psi\rangle\right\|\leqslant\varepsilon,
 \end{align}
where $|\Psi\rangle=|0^{m}\rangle\otimes|\psi\rangle$ and $|\psi\rangle$ is an arbitrary $N$-dimensional quantum state.
We may rewrite $W$ as 
\begin{align}\label{eq:defw}
W=\begin{pmatrix}\delta H&W_{1,2}\\W_{2,1}&W_{2,2}\end{pmatrix}
\end{align}
where $\delta H\in\mathbb{C}^{P\times P}$, $W_{1,2}\in\mathbb{C}^{N\times P}$, $W_{2,1}\in\mathbb{C}^{P\times N}$ and $W_{2,2}\in\mathbb{C}^{N\times N}$. Note that if $\|\delta H\|\leqslant\varepsilon$, $\tilde U$ is a $(m,\varepsilon)$-block-encoding of $H$. We have

\begin{align}\label{eq:wpsi}
W|\Psi\rangle=\begin{pmatrix}\delta H&W_{1,2}\\W_{2,1}&W_{2,2}\end{pmatrix}\begin{pmatrix}|\psi\rangle\\0\end{pmatrix}
=
\begin{pmatrix}\delta H|\psi\rangle\\W_{2,1}|\psi\rangle\end{pmatrix}
\end{align}
Combining Eq.~\eqref{eq:wpsi0} with Eq.~\eqref{eq:wpsi}, we have 
\begin{align}\label{eq:dh}
\|\delta H|\psi\rangle\|\leqslant\|W|\Psi\rangle\|\leqslant\varepsilon.
\end{align}
Because Eq.~\eqref{eq:dh} is applied for arbitrary $|\psi\rangle$, we have $\|\delta H\|\leqslant\varepsilon$. Therefore, $\tilde U$ is a $(m,\varepsilon)$-block-encoding to $H$.  We can now study the efficiency of block-encoding.

The actual circuit complexity depends on the form of $u_p$. We now proof Theorem.~\ref{th:bl1} which corresponds to $u_p\in\{\pm I,\pm X,\pm Y,\pm Z\}^{\otimes n}$.

\begin{proof}[Proof of Theorem.~\ref{th:bl1}]
With $n_\text{anc}$ ancillary qubits where $\log_2P\leqslant n_\text{anc}\leqslant O(P)$, $\tilde{\mathbb{G}}$ can be constructed with $O(P\log(1/\varepsilon))$ count and $O\left(P\frac{\log(n_{\text{anc}})\log(\log(n_{\text{anc}})/\varepsilon)}{n_{\text{anc}}}\right)$ depth of Clifford$+T$ gates. With $\Omega(\log_2P)\leqslant n_{\text{anc}}\leqslant O(Pn)$, Select$(H_x)$ can be constructed with $O(nP)$ count and $O\left(nP\frac{\log n_{\text{anc}}}{n_{\text{anc}}}\right)$ depth of Clifford$+T$ gates. Therefore, the total gate count of Select$(H_x)$ is $O(P(n+\log(1/\varepsilon)))$. For $\Omega(\log_2P)\leqslant n_{\text{anc}}\leqslant O(P)$, the circuit depth is 

\begin{align}
&O\left(P\frac{\log n_{\text{anc}}}{n_{\text{anc}}} \left(n+\log(\log(n_{\text{anc}})/\varepsilon\right))\right)\noindent\\
=&O\left(P\left(n+\log(1/\varepsilon)\right)\frac{\log n_{\text{anc}}}{n_{\text{anc}}} \right).
\end{align}
For $\Omega(P)\leqslant n_{\text{anc}}\leqslant O(Pn)$, the circuit depth for $\tilde{\mathbb{G}}$ is $O(\log P\log(\log(P)/\varepsilon))=O(\log n\log(\log n/\varepsilon))$, where we have used the assumption $P=O(\text{Poly}(n))$. Combining with circuit depth of Select$(H_x)$, the total circuit depth for block-encoding is

\begin{align}
&O\left(\left(\frac{n_{\text{anc}}\log(n)\log(\log (n)/\varepsilon)}{\log n_{\text{anc}}}+nP\right)\frac{\log n_{\text{anc}}}{n_{\text{anc}}} \right)\notag\\
=&O\left(\left(\frac{nP\log(n)\log(\log (n)/\varepsilon)}{\log (nP)}+nP\right)\frac{\log n_{\text{anc}}}{n_{\text{anc}}} \right)\notag\\
=&O\left(\left(nP\log(\log (n))+nP\log(1/\varepsilon)\right)\frac{\log n_{\text{anc}}}{n_{\text{anc}}} \right)\notag\\
=&\tilde O\left(nP\log(1/\varepsilon)\frac{\log n_{\text{anc}}}{n_{\text{anc}}} \right),
\end{align}
which completes the proof. 
\end{proof}

$\\$
\noindent\textbf{Sparse Boolean memory}

\noindent Recall that sparse Boolean memory performs the transformation $\text{Select}(B)|q\rangle_{\text{idx}}|z\rangle_{\text{wrd}}=|q\rangle_{\text{idx}}|z\oplus B(q)\rangle_{\text{wrd}}$, idx represents an $n$-qubit index register, wrd represents a $\tilde n$-qubit register, and there are most $s$ input digits $q$ satisfying $B(q)\neq0\cdots0$.
We define $q_k$ as the $k$th input digit with nonzero output, and $\mathcal{Q}_B\equiv\{q_1,q_2,\cdots,q_s\}$. 
In~\cite{Zhang.22}, we have developed a construction of SBM with $O(ns\tilde n)$ ancillary qubits. The result is as follows.

\begin{lemma}[Sec.III B in Supplemental Material of~\cite{Zhang.22}]\label{th:sbm_many}
With $O(ns\tilde n)$ ancillary qubits, Select$(B)$ can be realized with $O(ns\tilde n)$ count and $O(\log(ns\tilde n))$ depth of Clifford$+T$ gates.
\end{lemma}

Based on Lemma.~\ref{th:sbm_many}, we can obtain the gate complexity with intermediate number of ancillary qubits. The proof of Lemma.~\ref{th:sbm} is given as follows.

\begin{proof}[Proof of Lemma.~\ref{th:sbm}]
  Let $\text{wrd}_l$ be the $l$th qubit of the word register, and $z_l$ be the $l$th digit of $z$. So $|z\rangle_\text{wrd}=\prod_{l=1}^{\tilde n}|z_{l}\rangle_{\text{wrd}_{l}}$. $\text{Select}(B)$ can be separated into multiple Boolean functions applied at different words. 
 Let $B_{l}(q)$ be the $l$th digit of $B(q)$, and $B_{l_{\min}:l_{\max}}(q)\equiv B_{l_{\max}}(q)\cdots B_{l_{\min}+1}(q)B_{l_{\min}}(q)$. We define $\text{Select}(B_{l_{\min}:l_{\max}})$ as a unitary satisfying 
\begin{align}
&\text{Select}(B_{l_{\min}:l_{\max}})|q\rangle_{\text{idx}}\prod_{l=l_{\min}}^{l_{\max}}|z_l\rangle_{\text{wrd}_l}\\
=&|q\rangle_{\text{idx}}\prod_{l=l_{\min}}^{l_{\max}}|z_l\oplus B_l(q)\rangle_{\text{wrd}_l}.
\end{align}
For any $1=l_0< l_1<\cdots< l_{n'}=\tilde n+1$, it can be verified that 
\begin{align}\label{eq:B_a}
\text{Select}(B)=\prod_{r=1}^{n'}\text{Select}(B_{l_{r-1}:l_{r}-1}).
\end{align}
We also define $\text{Select}(B_{l})=\text{Select}(B_{l:l})$. For each $B_l$, we further define Boolean functions $B_{l,k_{\min}:k_{\max}}(q)=B_{l}(q)\wedge (k_{\min}\leqslant k\leqslant k_{\max})$ for $k_{\min}\leqslant k_{\max}$. For any $0=k_0<k_1<\cdots<k_{s'}=s$, it can be verified that  
\begin{align}\label{eq:B_b}
\text{Select}(B_l)=\prod_{j=1}^{s'}\text{Select}(B_{l,k_{j-1}:k_j-1}).
\end{align}

We first consider the construction with ancillary qubit number $O(ns)\leqslant n_{\text{anc}}\leqslant O(ns\widetilde n)$. In this case, we decompose Select$(B(q))$ with Eq.~\eqref{eq:B_a}. We let $d=\lfloor n_{\text{anc}}/(ns)\rfloor$ and $n'=\lceil\tilde n/d\rceil$, and 
\begin{eqnarray}
l_r= \left\{
\begin{array}{lll}
rd+1     &    & r< n'\\
\tilde n+1    &  & r=n'\\  
\end{array} \right. \ .
\end{eqnarray}
 According to Lemma.~\ref{th:sbm_many}, with $n_{\text{anc}}$ ancillary qubits, each $\text{Select}(B_{l_{r-1}:l_{r}-1})$ can be constructed with $O(nsd)$ count and $O(\log(nsd))=O(\log n_{\text{anc}})$ depth of Clifford$+T$ circuit. So the total gate count is $O(nsd)\times n'=O(ns\tilde n)$, and the total circuit depth is $O(\log(nsn_{\text{anc}}))\times n'=O\left(ns\tilde n\frac{\log n_{\text{anc}}}{n_{\text{anc}}}\right)$.

We then consider the construction with ancillary qubit number $O(n)\leqslant n_{\text{anc}}\leqslant O(ns)$. In this case, we first perform the decomposition $\text{Select}(B)=\prod_{r=1}^{\tilde n}\text{Select}(B_{l})$. Then, we decompose each $\text{Select}(B_{l})$ with Eq.~\eqref{eq:B_b}. We let $w=\lfloor m/n\rfloor$ and $s'=\lceil s/w\rceil$, and 
\begin{eqnarray}
k_j= \left\{
\begin{array}{lll}
jw     &    & j< s'\\
s    &  & j=n'\\  
\end{array} \right. \ .
\end{eqnarray}
According to Lemma.~\ref{th:sbm_many}, with $n_{\text{anc}}$ ancillary qubits, each $\text{Select}(B_{l,k_{j-1}:k_{j}-1})$ can be constructed with $O(nw)$ count and $O(\log(nw))=O(\log n_{\text{anc}})$ depth of Clifford$+T$ circuit. So each Select$(B_l)$ requires gate count $O(nw)\times s'=O(ns)$, and circuit depth $O(\log(n_{\text{anc}}))\times s'=O\left(ns\frac{\log n_{\text{anc}}}{n_{\text{anc}}}\right)$. In this case, we have $n'=\tilde n$ in Eq.~\eqref{eq:B_a}, so the total gate count and circuit depth of Select$(B(q))$ is $O(ns\tilde n)$ and $O\left(ns\tilde n\frac{\log n_{\text{anc}}}{n_{\text{anc}}}\right)$  respectively.
\end{proof}

\noindent\textbf{Data availability} Data sharing is not applicable to this article as no datasets were generated or analysed during the current study.

\noindent\textbf{Acknowledgement} This work is supported by the National Natural Science Foundation of China (Grant No. 12175003, 12247124 and 12361161602), NSAF (Grant No.~U2330201) and Project funded by China Postdoctoral Science Foundation (Grant No. 2023T160004).


\noindent\textbf{Competing interests} The authors declare that there are no competing interests. 

\noindent\textbf{Author Contributions} Xiao-Ming Zhang conceived the project. All authors  contributed in the preparation of the manuscript.

%

\vspace{1cm}
\onecolumngrid
\newpage

\begin{center}
{\bf\large Supplementary material}
\end{center}
\vspace{0.5cm}

\setcounter{secnumdepth}{3}  
\setcounter{equation}{0}
\setcounter{figure}{0}
\setcounter{table}{0}
\setcounter{section}{0}

\renewcommand{\theequation}{S-\arabic{equation}}
\renewcommand{\thefigure}{S\arabic{figure}}
\renewcommand{\thetable}{S-\Roman{table}}
\renewcommand\figurename{Supplementary Figure}
\renewcommand\tablename{Supplementary Table}
\newcommand\citetwo[2]{[S\citealp{#1}, S\citealp{#2}]}
\newcommand\citecite[2]{[\citealp{#1}, S\citealp{#2}]}

\newcolumntype{M}[1]{>{\centering\arraybackslash}m{#1}}
\newcolumntype{N}{@{}m{0pt}@{}}

\makeatletter \renewcommand\@biblabel[1]{[S#1]} \makeatother

\makeatletter \renewcommand\@biblabel[1]{[S#1]} \makeatother



\section{Supplementary Discussion 1: Circuit complexity lower bound}\label{app:lower}
\subsection{Capacity of quantum circuits with fixed gate count: proof of Lemma.~\ref{lm:CC}}\label{app:CC}

Without loss of generality, we require that all ancillary qubits are initialized at trivial state, $|0\rangle$, and uncomputed at the end of the quantum circuits. The quantum circuit with $C$ quantum gates can generally be expressed as 
\begin{align}\label{eq:G}
G=\prod_{j=1}^CG(\theta_j,a_j,b_j),
\end{align}
where $G(\theta_j,a_j,b_j)$ is the two-qubit gate $G(\theta_j)\in\mathcal{G}_{\text{ele}}$ applied at qubits with label $a_j$ and $b_j$, which can be either data or ancillary qubits. $\theta_j$ is the label of the elementary gate. Our proof is based on Algorithm.~\ref{alg:unique_gates} which specify the parameters in Eq.~\eqref{eq:G}. We denote the set containing all unitaries that can be constructed with Algorithm.~\ref{alg:unique_gates} as $\mathcal{G}_{\text{alg1}}$. All unitaries in $\mathcal{G}_C$ (i.e. all ancillary qubits are uncomputed) can be constructed with Algorithm.~\ref{alg:unique_gates}. So we have $\mathcal{G}_{C}\subset\mathcal{G}_{\text{alg1}}$, and $|\mathcal{G}_{C}|\leqslant |\mathcal{G}_{\text{alg1}}|$.  In below, we discuss the number of elements in $\mathcal{G}_{\text{alg1}}$.

Algorithm.~\ref{alg:unique_gates} contains two stages. In the first stage, we determine $a_j$ and $b_j$ for each $j$, which represent the ``skeleton'' of the circuit without specifying what exact elementary gate is applied. In the second stage, we determine $\theta_j$ for each $j$, i.e. fill in the skeleton with specific elementary gates in $\mathcal{G}_{\text{ele}}$.

We first consider stage $1$. We define $\mathcal{D}$ as a set containing all data qubits. For the first gate $j=1$, we may choose to introduce $0, 1$ or $2$ ancillary qubits in which the first gate is applied at. In case introducing $0$ ancillary qubit, we choose both $a_1$ and $b_1$ among $\mathcal{D}$, which accounts for $n(n-1)$ possible choices. In case introducing $1$ ancillary qubit, we label the ancillary qubit being introduced as $A_1$, and set $a_1=A_1$ (or $b_1=A_1$). We then choose $b_1$ (or $a_1$) among $\mathcal{D}$. This case accounts for $2n$ possible choices. In case introducing $2$ ancillary qubits, we label them as $A_1$ and $A_2$ respectively. We set $a_1=A_1, b_2=A_2$, or $a_1=A_2, b_2=A_1$. So this case accounts for $2$ possible choices. Taking all cases into consideration, there are totally $n(n-1)+2n+2=O(n^2)$ possible choices for the first gate. After this step, we set $n_{\text{anc}}$ as the number of ancillary qubits that has been introduced ($n_{\text{anc}}=0,1$ or $2$).

The subsequent gates ($j>1$) are more involved. For the $j$th gate, we define $\mathcal{D}_j$ as the set containing all data qubits and all ancillary qubits that have been introduced. Note that ancillary qubits in $\mathcal{D}_j$ have been labeled as $A_1, A_2,\cdots,A_{n_{\text{anc}}}$ respectively.  Similarly, for the $j$th gate, we can introduce $0, 1$ or $2$ extra ancillary qubits in which the gate is applied at. We note that $|\mathcal{D}_j|\leqslant (n+2(j-1))$. So if we introduce $0$ extra ancillary qubits, there are totally possible $|\mathcal{D}_j|(|\mathcal{D}_j|-1)$ choices for $(a_j, b_j)$. If we introduce 1 extra ancillary qubit, we set $n_{\text{anc}}\leftarrow n_{\text{anc}}+1$, and label the extra qubit as $A_{n_{\text{anc}}}$. 
In this case, we set $a_j=A_{n_{\text{anc}}}$ (or $b_j=A_{n_{\text{anc}}}$), and choose $b_j$ among all qubits in $\mathcal{D}_j$. There are totally $O(|\mathcal{D}_j|)$ possible choices. If we introduce $2$ extra ancillary qubits, we set $n_{\text{anc}}\leftarrow n_{\text{anc}}+2$, and label the them as $A_{n_{\text{anc}}-1}$ and $A_{n_{\text{anc}}}$ respectively. In this case, we set $a_j=A_{n_{\text{anc}}-1}, b_2=A_{n_{\text{anc}}}$, or $a_1=A_{n_{\text{anc}}}, b_2=A_{n_{\text{anc}}-1}$. So there are $2$ possible choices. Taking all cases into considerations, there are totally $O(|\mathcal{D}_j|^2)$ possible choices for the $j$th gate. It can be verified iteratively that $|\mathcal{D}_j|\leqslant n+2(j-1)$, so there are totally $O((n+j)^2)$ choices.

Summing up all possible choices from $j=1$ to $j=C$, there are totally 
\begin{align}
\prod_{j=1}^CO((n+j)^2)=O((n+C)^C)
\end{align}
possible choices of all labels $a_j, b_j$ in the first stage.

The second stage is to determine $\{\theta_j\}$, which is simpler. We choose each $G(\theta_j)$ among $\mathcal{G}_{\text{ele}}$ independently, and there are totally $O(g^C)$ possible choices. 

Combining both stages, in Algorithm.~\ref{alg:unique_gates}, there are totally $O\left(((C+n)g)^C\right)$ possible choices of the parameters, i.e. $|\mathcal{G}_{\text{alg1}}|=O\left(((C+n)g)^C\right)$. 
It can be verified that all unitaries in $\mathcal{G}_{C}$ can be constructed with Algorithm.~\ref{alg:unique_gates}, so we have 
\begin{align}
|\mathcal{G}_{C}|\leqslant |\mathcal{G}_{\text{alg}1}|=O\left(((C+n)g)^C\right).
\end{align}
 Because $g=O(1)$, we have 
\begin{align}\label{eq:logc}
\log|\mathcal{G}_{C}|=O\left(C\log(C+n)\right),
\end{align}
which completes the proof.

\begin{algorithm} [H]
\caption{Pseudo code for determining parameters in Eq.~\eqref{eq:G}  \label{alg:unique_gates}}
\begin{algorithmic}[1]
\STATE set $\mathcal{D}_1\leftarrow\mathcal{D}$ and $n_{\text{anc}}=0$  \quad $\#$ Beginning of stage 1
\STATE \textbf{for} $j=1$ \textbf{to} $C$:
\STATE \quad\textbf{choose} Extra$\_$Anc=0 \textbf{or} $1$  \textbf{or} $2$  
\STATE \quad\quad $n_{\text{anc}}\leftarrow n_{\text{anc}}+$Extra$\_$Anc
\STATE \quad\textbf{if} Extra$\_$Anc=0:
\STATE \quad\quad  \textbf{choose} $a_j$ among $\mathcal{D}_j$
\STATE \quad\quad  \textbf{choose} $b_j$ among $\mathcal{D}_j-\{a_j\}$
\STATE \quad\quad $\mathcal{D}_{j+1}\leftarrow\mathcal{D}_j$

\STATE \quad\textbf{else if} Extra$\_$Anc=1:
\STATE \quad\quad \textbf{choose} $a_j=A_{n_\text{anc}}$ \textbf{or} $b_j=A_{n_\text{anc}}$
\STATE \quad\quad \textbf{if} $a_j=A_{n_\text{anc}}$:
\STATE \quad\quad\quad  \textbf{choose} $b_j$ among $\mathcal{D}_j$
\STATE \quad\quad \textbf{else if} $b_j=A_{n_\text{anc}}$:
\STATE \quad\quad\quad  \textbf{choose} $a_j$ among $\mathcal{D}_j$
\STATE \quad\quad \textbf{end if} 
\STATE \quad\quad $\mathcal{D}_{j+1}\leftarrow\mathcal{D}_j\bigcup\{A_{n_\text{anc}}\}$
\STATE \quad\quad $n_\text{anc}\leftarrow n_\text{anc}+1$

\STATE \quad\textbf{else if} Extra$\_$Anc=2:
\STATE \quad\quad  \textbf{choose} $(a_j,b_j)=(A_{n_\text{anc}-1}, A_{n_\text{anc}})$ \textbf{or} $(a_j,b_j)=(A_{n_\text{anc}}, A_{n_\text{anc}-1})$
\STATE \quad\quad $\mathcal{D}_{j+1}\leftarrow\mathcal{D}_j\bigcup\{A_{n_\text{anc}-1}, A_{n_\text{anc}}\}$

\STATE  \quad\textbf{end if}
\STATE \textbf{end for}

\STATE \textbf{for} $j=1$ \textbf{to} $C$:   \quad $\#$ Beginning of stage 2
\STATE\quad  \textbf{choose} $G(\theta_j)$ among $\mathcal{G}_{\text{ele}}$
\STATE \textbf{end for}

\end{algorithmic}
\end{algorithm}

\subsection{Capacity of quantum circuits with fixed depth and qubit number: proof of Lemma.~\ref{lm:cancd} }\label{app:cancd}

We then study the upper bound of the total number of unique quantum circuits with circuit depth $D$ and ancillary qubit number $n_{\text{anc}}$. 
With a fixed qubit number $n+n_{\text{anc}}$, the quantum circuit may be constructed in the following way. In the first stage, at each layer, we partition $(n+n_{\text{anc}})$ qubits into $(n+n_{\text{anc}})/2$ subsets (assuming $(n+n_{\text{anc}})$ is even), each of which contains two qubits. There are totally $O((n+n_{\text{anc}})^{2D})$ possible ways of partitions. In the second stage, we fill each subset with an elementary gate in $\mathcal{G}_{\text{ele}}$. There are totally $g$ choices for each subset, so there are totally $O(g^{(n+n_{\text{anc}})D/2})$ possible choices in this stage. 
$\mathcal{G}'_{n_{\text{anc}},D}$ is a subset of quantum circuits constructed in the process above. So we have 
\begin{align}
\left|\mathcal{G}'_{n_{\text{anc}},D}\right|=O\left((n+n_{\text{anc}})^{2D}\times g^{(n+n_{\text{anc}})D/2}\right).
\end{align}
Because $g=O(1)$, we have
\begin{align}\label{eq:logd}
\log\left|\mathcal{G}'_{n_{\text{anc}},D}\right|=O\left(D\log(n+n_{\text{anc}})+D(n+n_{\text{anc}})\right)=O(D(n+n_{\text{anc}})),
\end{align}
which completes the proof. 

$\left|\mathcal{G}'_{n_{\text{anc}},D}\right|$ may be tighten with a restriction similar to Lemma.~\ref{lm:CC}, i.e. acting trivially at the ancillary subspace. But the current result is sufficient for our analysis.

\subsection{Circuit depth lower bound for SAIM: proof of Theorem.~\ref{th:lbsm}}\label{app:lbsm}
To begin with, we introduce the following lemma about the number of unique matrices required to ``cover'' another set of matrices.  

\begin{lemma}\label{lm:mono}
Given a set of matrices $\mathcal{M}$ satisfying the following: If $M_1,M_2\in\mathcal{M}$ and $M_1\neq M_2$, then, $\|M_1-M_2\|\geqslant 2\varepsilon$. We suppose there is another set of matrices $\mathcal{M}'$ satisfying the following: If $M\in\mathcal{M}$, there exist $M'\in\mathcal{M}$ such that $\|M-M'\|<\varepsilon$. Then, we have $|\mathcal{M}'|\geqslant |\mathcal{M}|$.
\end{lemma}

\begin{proof}
We suppose $\|M_1-M_2\|\geqslant 2\varepsilon$ and $\|M_1-M'\|<\varepsilon$. According to triangular inequality, we have $\|M_2-M'\|> 2\varepsilon-\varepsilon=\varepsilon$. Therefore, any $M'\in\mathcal{M}'$ can approximate at most one of the elements in $\mathcal{M}$ to accuracy $\varepsilon$. Therefore, to approximate all elements in $\mathcal{M}$, we have $|\mathcal{M}'|\geqslant|\mathcal{M}|$.
\end{proof}

We denote $\mathcal{O}_{\text{val}}$ and $\mathcal{O}_{\text{find}}$ as the sets containing all $O_H$  and $O_F$ with different parameters of the sparse matrix. To begin with, we estimate the lower bound of $|\mathcal{O}_{\text{val}}|$ and $|\mathcal{O}_{\text{find}}|$.
We notice that $|\mathcal{O}_{\text{val}}|$ and $|\mathcal{O}_{\text{find}}|$ increases with sparsity $s$ and number of digits $d$, so we can only consider the case with $s=1$ and $d=1$ to give a lower bound. In this case, we have $|\mathcal{O}_{\text{val}}|=|\mathcal{O}_{\text{find}}|=\Omega(N!)$.  

For each pair of $O_{H_1}, O_{H_2}\in\mathcal{O}_{\text{val}}$ such that $O_{H_1}\neq O_{H_2}$, there exist input $|x,y\rangle$ in computational basis satisfying $\|O_{H_1}|x,y\rangle- O_{H_2}|x,y\rangle\|=2$. So we have $\|O_{H_1}-O_{H_2}\|\geqslant2$. According to Lemma.~\ref{lm:mono}, to approximate all elements in $\mathcal{O}_{\text{val}}$ to accuracy $1$ with $\mathcal{G}_C$ or $\mathcal{G}'_{n_{\text{anc},D}}$, it is required that  $|\mathcal{G}_{C}|\geqslant\Omega(N!)$ and $|\mathcal{G}'_{n_{\text{anc}},D}|\geqslant\Omega(N!)$. Combining with Eq.~\eqref{eq:logc} for Lemma.~\ref{lm:CC} and Eq.~\eqref{eq:logd} for Lemma.~\ref{lm:cancd}, we have 

\begin{align}
C\log(n+C)&=\Omega(\log (N!)),\\
D(n+n_{\text{anc}})&=\Omega(\log(N!)).
\end{align}
Notice that $C=\Omega(n)$, we have
\begin{align}
C\log C=\Omega(N\log N),\\
(n+n_{\text{anc}})D=\Omega(Nn),
\end{align}
which are equivalent to $C=\Omega(N)$ and $(n+n_{\text{anc}})D=\Omega(2^n)$ as stated in Theorem.~\ref{th:lbsm}.

\subsection{Circuit depth lower bound for sparse matrix block-encoding: proof of Theorem.~\ref{th:hd_spbe}}\label{app:lbblsm}
We first introduce a lemma similar to Lemma.~\ref{lm:mono} as follows.
\begin{lemma}\label{lm:mono_2}
Given a set of matrices $\mathcal{M}$ satisfying the following: 
(1) The matrix is normalized, i.e. $\|M\|=1$ for any $M\in\mathcal{M}$.
(2) If $M_1,M_2\in\mathcal{M}$ and $M_1\neq M_2$, then, $\|M_1-M_2\|\geqslant 4\varepsilon$ for a sufficiently small $\varepsilon$. 

We suppose there is another set of unitaties $\mathcal{U}$ with fixed dimension satisfying the following: For any $M\in\mathcal{M}$, there exist $U\in\mathcal{U}$ such that $U$ is a $(\alpha,n_{\text{anc}},\varepsilon)$-block-encoding of $M$ for some $\alpha>0$. Then, we have $|\mathcal{U}|\geqslant |\mathcal{M}|$.
\end{lemma}

\begin{proof}
We consider $\|M_1-M_2\|\geqslant 4\varepsilon$ for some $M_1,M_2\in\mathcal{M}$. We suppose $U$ is a $(\alpha_1,n_{\text{anc}},\varepsilon)$-block-encoding of $M_1$ for some $\alpha_1>0$. According to definition, let $M'\equiv\alpha_1\left(\langle0^{n_{\text{anc}}}|\otimes\mathbb{I}_{N}\right) U\left(|0^{n_{\text{anc}}}\rangle\otimes\mathbb{I}_{N}\right)$, we have $\|M'-M_1\|\leqslant\varepsilon$. According to triangular inequality, have $\big|\|M'\|-1\big|\leqslant\varepsilon$, i.e. $\|M'\|\in[1-\varepsilon,1+\varepsilon]$.

Now, we assume that $U$ is also a $(\alpha_2,n_{\text{anc}},\varepsilon)$-block-encoding of $M_1$ for some $\alpha_2>0$. In other words, we have 
\begin{align}\label{eq:contr}
\|\beta M'-M_2\|\leqslant\varepsilon,
\end{align}
 where $\beta=\alpha_2/\alpha_1$. Similarly, we have $\big|\beta\|M'\|-1\big|\leqslant\varepsilon$, i.e. $\beta\in\left[\frac{1-\varepsilon}{1+\varepsilon},\frac{1+\varepsilon}{1-\varepsilon}\right]$. On the other hand, according to triangular inequality, we have 
 
 $$\|\beta M'-M_2\|\geqslant \|M_1-M_2\|-\|M_1-\beta M'\|\geqslant 4\varepsilon-\|(M_1-M')+(1-\beta)M'\|\geqslant4\varepsilon-\varepsilon-(1-\beta)\|M'\|\geqslant\varepsilon.$$
 which is contradicted to Eq.~\eqref{eq:contr}. In other words,  $U$ can not be a block-encoding of $M_2$ for accuracy $\varepsilon$. Then, with a similar argument to Lemma.~\ref{lm:mono}, we have $|\mathcal{U}|\geqslant|\mathcal{M}|$.
 
\end{proof}

 The circuit depth lower bound of sparse matrix block-encoding can be obtained based on  Lemma.~\ref{lm:mono_2}.
To obtain a lower bound can just consider a special matrix form, the diagonal matrix 
\begin{align}\label{eq:spa}
A(\bm{v})=\sum_{j=0}^{N-1}v_j|j\rangle\langle j|,
\end{align}
where $\bm{v}=[v_0,v_1,\cdots,v_{N-1}]$. We define 
\begin{align}
\mathcal{A}\equiv\{A(\bm{v})|v_j=-1 \quad\text{or}\quad v_j=1\}.
\end{align}
It can be verified that $|\mathcal{A}|=2^N$, and for any $A\in\mathcal{A}$, we have $\|A\|=1$. Moreover, for $A_1,A_2\in\mathcal{A}$ and $A_1\neq A_2$, we have $\|A(\bm{v_1})-A(\bm{v_2})\|\geqslant 2$. Combining  Lemma.~\ref{lm:CC},~\ref{lm:cancd} with Lemma.~\ref{lm:mono_2}, we have

\begin{align}
C\log(C+n)&=\Omega(\log (2^N)),\label{eq:spc2}\\
D(n+n_{\text{anc}})&=\Omega(\log (2^N)).
\end{align}
Equivalently, we have $(n+n_{\text{anc}})D=\Omega(N)$ and $C=\Omega(N^{\alpha})$ for arbitrary $\alpha\in(0,1)$.

\subsection{Circuit complexity lower bound for quantum state preparation}
In Section 4.5.4 of~\cite{Nielsen.02}, it has been shown that to achieve accuracy $\varepsilon$ for an $n$-qubit quantum state, at least $\Omega\left(\frac{1}{\varepsilon^{2^{n+1}}-1}\right)$ unique quantum circuits are required [Eq.(4.83)]. According to Lemma.~\ref{lm:CC},~\ref{lm:cancd}, we have the following circuit complexity for quantum state preparation.

\begin{theorem}
Given an arbitrary finite two-qubit elementary gate set $\mathcal{G}_{\text{ele}}$.
Let $n_{\text{anc}}$, $D$ and $C$ be the number of ancillary qubits, circuit depth and total number of gates in $\mathcal{G}_{\text{ele}}$ required to approximate an arbitrary quantum state to accuracy $\varepsilon<1$. Then, we have $(n+n_{\text{anc}})D=\Omega(2^n\log(1/\varepsilon))$ and $C=\Omega((2^n\log_2(1/\varepsilon))^\alpha)$ for arbitrary $\alpha\in(0,1)$.
\end{theorem}
\begin{proof}
For Lemma.1, we have
\begin{align}
\log|\mathcal{G}_C|&=O(C\log(C+n)).
\end{align}
To satisfy the requirement of approximating arbitrary state, we have $|\mathcal{G}_C|=\Omega\left(\frac{1}{\varepsilon^{2^{n+1}}-1}\right)$, and equivalently,
\begin{align}
\log|\mathcal{G}_C|&=\Omega(2^n\log(1/\varepsilon)).
\end{align}
Therefore, we should have $K_2 C\log(C+n)>K_1 2^n\log(1/\varepsilon)$ for some constant $K_1,K_2$. Equivalently, we have
\begin{align}
C\log(C+n)=\Omega(2^n\log(1/\varepsilon)).
\end{align}
Because the number of gates is typically larger than the number of qubits, i.e. $C>n$, so we have $C\log(2C)>C\log(C+n)$. Therefore
\begin{align}
C\log(C)=\Omega(2^n\log(1/\varepsilon)).
\end{align}

For any $\alpha\in(0,1)$, we have $C^{1/\alpha}=C\cdot C^{1/\alpha-1}>K_3C\log(C)$ for some $K_3$. This is because $\beta\equiv1/\alpha-1>0$, and $C^{\beta}$ increases faster than $\log C$. So we have $C^{1/\alpha}=\Omega(C\log C)=\Omega(2^n\log(1/\varepsilon))$, which is equivalent to 
\begin{align}
C=\Omega(C\log C)=\Omega((2^n\log(1/\varepsilon))^{\alpha}).
\end{align}
\end{proof}

\subsection{Generalization to quantum channels}
Our result can be readily generalized to quantum channels. More specifically, we are given an elementary set of quantum channels
\begin{align}\label{eq:chan1}
\mathcal{E}_{\text{ele}}^{\text{(channel)}}=\left\{\mathcal{E}_1,\mathcal{E}_2,\cdots,\mathcal{E}_g\right\},
\end{align}
where $\mathcal{E}_j$ can be an arbitrary operations applied at two-qubit systems. This includes unitary, measurement, and corresponding feedback control, etc.  We also restrict that $\left|\mathcal{E}_{\text{ele}}^{\text{(channel)}}\right|=O(1)$, i.e. there are constant number of elementary quantum channels. With this definition, we have the following results.

\begin{lemma}\label{lm:channel}
Let $\mathcal{G}^{\text{(channel)}}_{C}$ be the set containing all $n$-qubit quantum channels that can be constructed with $C$ elementary channels in $\mathcal{G}^{(channel)}_{\text{ele}}$ defined in Eq~\eqref{eq:chan1} above. Then, we have $\log\big|\mathcal{G}^{\text{(channel)}}_{C}\big|= O\left((C\log(C+n))\right)$, even with unlimited ancillary qubit number. 

Let $\mathcal{G}^{\text{(channel)}'}_{n_{\text{anc}},D}$ be the set containing all unitaries that can be constructed with $n_{\text{anc}}$ ancillary qubits and $D$ circuit depth. Then, we have $\log\big|\mathcal{G}^{\text{(channel)}'}_{n_{\text{anc}},D}\big|=O\left(D(n+n_{\text{anc}})\right)$.
\end{lemma}

\begin{proof}
Similar to the unitary case, a quantum channel constructed with $C$ elementary channels in $\mathcal{G}_{\text{ele}}^{\text{(channel)}}$ can always be described as
\begin{align}\label{eq:G}
\mathcal{E}=\prod_{j=1}^C\mathcal{E}(\theta_j,a_j,b_j),
\end{align}
where $\mathcal{E}(\theta_j,a_j,b_j)$ is the two-qubit quantum channel $\mathcal{E}(\theta_j)\in\mathcal{G}^{\text{(channel)}}_{\text{ele}}$ applied at qubits with label $a_j$ and $b_j$. $\log\big|\mathcal{G}^{\text{(channel)}}_{C}\big|$ and $\log\big|\mathcal{G}^{\text{(channel)}'}_{n_{\text{anc}},D}\big|$ can then be estimated with the same arguments in the proof of Lemma.~\ref{lm:CC} and Lemma.~\ref{lm:cancd}.
\end{proof}

The reason why Lemma.~\ref{lm:channel} have the same complexity to unitary case is that we have restricted that the elementary operations are still local, i.e. applied nontrivially on at most $2$ qubits. Based on Lemma.~\ref{lm:channel}, we can also analysis lower bound about the circuit complexity with local quantum channels. Take the SAIM as an example, we have the following result.

\begin{theorem}\label{th:gen}
Given an arbitrary finite two-qubit quantum channel sets $\mathcal{G}^{\text{(channel)}}_{\text{ele}}$.
Let $n_{\text{anc}}$, $D$ and $C$ be the number of ancillary qubits, circuit depth and total number of gates in $\mathcal{G}^{\text{channel}}_{\text{ele}}$ required to approximate SAIM in Eq.~\eqref{eq:ofh} with any accuracy $\varepsilon<1$. Then, we have $(n+n_{\text{anc}})D=\Omega(2^nn)$ and $C=\Omega(2^{n})$.
\end{theorem}

The proof is straightforward by following the same process in Sec.~\ref{app:lbsm}. 
An open question is about the capacity when global quantum channels are allowed. For example, one may perform (local) unitaries conditioned on multiple measurement outcomes. In this case, the elementary operations can no longer be described by a two-qubit quantum channel. 

\section{Supplementary Discuss 2: Block-encoding for sparse  matrix}\label{app:sp_be}

The circuit complexity lower bound for block-encoding sparse matrix has been discussed in Theorem.~\ref{th:hd_spbe}.
In this section, consider the construction block encoding of a sparse matrix $A$ with at most $s=O(1)$ nonzero elements at each row and column. To facilitate the discussion, we consider the normalized matrices $\|A\|=1$ and set the normalization factor for block-encoding as $\alpha=1$.

\subsection{Sparse state preparation}\label{app:sp_sps}
Here, we consider the construction of a $n$-dimensional, $s$-sparse quantum state 
\begin{align}
|\psi\rangle=\sum_{k=0}^{s-1}\psi_k|q_k\rangle.
\end{align}
where $q_k$ is an $n$-qubit computational basis. We have the following result.

\begin{theorem}\label{th:sps}
Given $n_{\text{anc}}$ ancillary qubits where $\Omega(n)\leqslant n_{\text{anc}}\leqslant O(ns\log s)$. An arbitrary (controlled-)quantum state preparation unitary with accuracy $\varepsilon$ can be constructed with $O(s(n\log s+\log(1/\varepsilon)))$ count and $\tilde{O}\left( ns\log s\log(1/\varepsilon)\frac{\log n_{\text{anc}}}{n_{\text{anc}}}\right)$ depth of Clifford$+T$ gates. Here, $\tilde O$ suppresses doubly logarithmic factors of $n_{\text{anc}}$.
\end{theorem}

\begin{proof}
We introduce two registers. Register $A$ contains $\lceil \log_2s\rceil$ qubits, and register $B$ contains $n$ qubits. They are initialized to $|0^{\lceil \log_2s\rceil}\rangle$ and $|0^n\rangle$ respectively.
In the first step, we perform the transformation 
\begin{align}
|0^{\lceil \log_2s\rceil}\rangle_A|0^n\rangle_B\longrightarrow\sum_{k=0}^{s-1}\tilde\psi_k|k\rangle_A|0\rangle_B.
\end{align}
for some $\left\|\sum_{k=0}^{s-1}\tilde\psi_k|k\rangle-\sum_{k=0}^{s-1}\psi_k|k\rangle\right\|\leqslant\varepsilon$.
According to Theorem.~\ref{th:sp_main}, this step requires $O(s\log(1/\varepsilon))$ count and $O\left(s\frac{\log(m_1)\log(\log(m_1)/\varepsilon}{m_1}\right)$ depth of Clifford$+T$ gates, provide $m_1$ ancillary qubits satisfying  $\Omega(\log_2 s)\leqslant m_1\leqslant O(s)$. Because $s\leqslant 2^n$, this is also equivalent to $\Omega(n)\leqslant m_1\leqslant O(s)$.
In the second step, we perform the transformation
\begin{align}\label{eq:sp_3}
\sum_{k=0}^{s-1}\psi_k|k\rangle_A|0\rangle_B\longrightarrow\sum_{k=0}^{s-1}\psi_k|k\rangle_A|q_k\rangle_B.
\end{align}
Eq.~\eqref{eq:sp_3} can be realized by a select oracle for Pauli strings with $A$ as index register and $B$ as word register. According to Theorem.~\ref{th:sop}, this step requires $O(sn)$ count and $O\left(sn\frac{\log m_2}{m_2}\right)$ depth of Clifford$+T$ gates, provide $m_2$ ancillary qubits satisfying $\Omega(n)\leqslant m_2\leqslant O(ns)$. Finally, we perform the transformation 
\begin{align}\label{eq:sp_4}
\sum_{k=0}^{s-1}\psi_k|k\rangle_A|q_k\rangle_B\rightarrow\sum_{k=0}^{s-1}\psi_k|0\rangle_A|q_k\rangle_B,
\end{align}
which can be realized by a sparse Boolean memory with $B$ as index register and $A$ as word register. According to Lemma.~\ref{th:sbm}, this step requires $O(ns\log s)$ count and $O\left(ns\log s\frac{\log m_3}{m_3}\right)$ depth of Clifford$+T$ gates, provide $m_3$ ancillary qubits satisfying $\Omega(n)\leqslant m_3\leqslant O(ns\log s)$. After this step, the target state is then obtained by tracing out register $A$. By combing the cost of each step, the total circuit complexity stated in Theorem.~\ref{th:sps} can be obtained.

\end{proof}

\subsection{Construction of sparse matrix block-encoding }
\begin{theorem}\label{th:bl2}
Given $n_{\text{anc}}$ ancillary qubits where $\Omega(n)\leqslant n_{\text{anc}}\leqslant O(Nns\log s)$, the $(\alpha,n_{\text{anc}},\varepsilon)$-block-encoding of $s$-sparse matrix $A\in\mathbb{C}^{N\times N}$ with $\alpha=\|A\|_{F}$ can be constructed with $O(Ns(n\log s+\log(1/\varepsilon))$ count and $\tilde O\left(Nns\log s\log(1/\varepsilon)\frac{\log n_{\text{anc}}}{n_{\text{anc}}}\right)$ of Clifford$+T$ gates. Here, $\|A\|_{F}$ is the Frobenius norm of $A$, and $\tilde O$ suppresses doubly logarithmic factors of $n_{\text{anc}}$.
\end{theorem}

\begin{proof}
The block-encoding of a matrix can be realized with state preparation and controlled state preparation~\cite{Clader.22}. 
 We define $A_{j,k}$ as the element at the $j$th row and the $k$th column, i.e.  $A_{j,k}=\langle j|A|k\rangle$. We define quantum states
\begin{align}
|A\rangle&=\sum_{j=0}^{N-1}\frac{\|A_{j,\cdot}\|_F}{\|A\|_F}|j\rangle , \quad
|A_j\rangle=\sum_{k=0}^{N-1}\frac{A_{j,k}^*}{\|A_j,\cdot\|_F}|k\rangle. 
\end{align}
Note that $|A_j\rangle$ is a $s$-sparse quantum state, while $|A\rangle$ is a non-sparse state. We introduce state preparation unitaries $G_{A}$ and $G_{A_j}$ satisfying $G_{A}|0^n\rangle=|A\rangle$, and $G_{A_j}|0^n\rangle=|A_j\rangle$ respectively. We then define
\begin{align}
U_R=\sum_{j}|j\rangle\langle j|\otimes G_{A_j},\quad
U_L=\mathbb{I}_{N}\otimes G_{A}.
\end{align}
We use three subroutines satisfying the following
\begin{align}
\text{SWAP}|j\rangle|k\rangle&=|k\rangle|j\rangle,\\
U_{L}|k\rangle|0^{n}\rangle&=|k\rangle|A\rangle,\\
U_R|j\rangle|0^{n}\rangle&=|j\rangle|A_j\rangle.\label{eq:ur}
\end{align}
Eq.~\eqref{eq:ur} is equivalent to 
\begin{align}
\langle j|\langle0^{n}|U_R^\dag&=\langle j|\langle A_j|.
\end{align}
So we have 
\begin{align}
&\langle j|\langle 0^n|U_R^\dag \;\text{SWAP}\;U_L|k\rangle|0^{n}\rangle=\langle j|A\rangle\langle A_j|k\rangle.
=\frac{\|A_{j,\cdot}\|_F}{\|A\|_F}\frac{A_{j,k}}{\|A_{j,\cdot}\|_F}=\frac{A_{j,k}}{\|A\|_F}
\end{align}

Therefore $U_R^\dag \;\text{SWAP}\; U_L$ is a block-encoding to $A$ with normalization factor $\alpha=\|A\|_F$. Suppose we have two unitaries $\tilde U_R$ and $\tilde U_L$ satisfying $\| \tilde U_R-U_R\|\leqslant\varepsilon/2$ and $\| \tilde U_L-U_L\|\leqslant\varepsilon/2$. It can be verified that  $\tilde U_R^\dag \;\text{SWAP}\; \tilde U_L$ is a $(\|A\|_F,n,\varepsilon)$-block-encoding to $A$, and $\alpha=\|A\|_F\geqslant \|A\|$.
Based on Theorem.~\ref{th:sps} for sparse state preparation, the circuit complexity of $U_R$ can be estimated in the similar way to Theorem.~\ref{th:sop} for Select$(H_x)$. The circuit complexity for  $U_R^\dag \;\text{SWAP}\; U_L$ then follows readily.

\end{proof}

\end{document}